\documentclass[a4paper,11pt]{article}
\usepackage[utf8x]{inputenc}
\usepackage{times}
\usepackage{fullpage,epsfig}
\usepackage{ifpdf}
\usepackage{amsthm}
\usepackage{sgame}
\usepackage{color}
\ifpdf
  \usepackage[pdftex]{hyperref}
\else
  \usepackage[hypertex]{hyperref}
\fi
\usepackage{amsmath, amssymb}

\theoremstyle{plain}
\newtheorem{theorem}{Theorem}[section]
\newtheorem{lemma}{Lemma}[section]
\newtheorem{cor}{Corollary}[section]
\newtheorem{prop}{Proposition}[section]

\theoremstyle{definition}
\newtheorem{definition}{Definition}[section]
\newtheorem{algorithm}{Algorithm}[section]

\theoremstyle{remark}
\newtheorem{remark}{Remark}[section]

\newcommand{\x}{{\mathbf x}}
\newcommand{\y}{{\mathbf y}}
\newcommand{\z}{{\mathbf z}}
\newcommand{\w}{{\mathbf w}}
\newcommand{\0}{{\mathbf 0}}
\newcommand{\p}{{\mathbf p}}
\newcommand{\m}{{\mathbf m}}

\newcommand{\G}{\mathcal{G}}

\newcommand{\OO}{{\mathcal O}}
\newcommand{\Dir}{\mathcal E}
\newcommand{\Time}{\mathcal T}
\newcommand{\A}{\mathcal{A}}

\newcommand{\cw}{{\sf CW}}

\newcommand{\tm}{t_\text{\rm mix}}

\newcommand{\tr}{{t_\text{\rm rel}}}
\newcommand{\pimin}{\pi_{\text{\rm min}}}
\newcommand{\tv}[1]{\left\|#1\right\|_{\rm TV}}

\newcommand{\pmt}[3]{t_{#1}^{#2}(#3)}

\newcommand{\sign}[1]{{\sf sign}\left(#1\right)}

\newcommand{\Prob}[2]{\mathbf{P}_{#1} \left( #2 \right)}
\newcommand{\Expec}[2]{\mathbf{E}_{#1} \left[ #2 \right]}

\newcommand{\tauset}[1]{\tau_{\scriptscriptstyle #1}}

\allowdisplaybreaks

\begin{document}

\raggedbottom

\title{Metastability of Asymptotically Well-Behaved Potential Games}

\author{Diodato Ferraioli\thanks{DIAG, Sapienza Universit\`a di Roma, Italy. E-mail: {\tt ferraioli@dis.uniroma1.it}.} \and Carmine Ventre\thanks{School of Computing, Teesside University, UK. E-mail: {\tt C.Ventre@tees.ac.uk}.}}
\date{}

\setcounter{page}{0}
\maketitle
\thispagestyle{empty}

\begin{abstract}
One of the main criticisms to game theory concerns the assumption of full rationality. Logit dynamics is a decentralized algorithm in which a level of irrationality (a.k.a. ``noise'') is introduced in players' behavior. In this context, the solution concept of interest becomes the logit equilibrium, as opposed to Nash equilibria. Logit equilibria are
distributions over strategy profiles that possess several nice properties, including existence and uniqueness. However, there are games in which their computation may take time exponential in the number of players. We therefore look at an approximate version of logit equilibria, called \emph{metastable distributions}, introduced by Auletta et al. \cite{afppSODA12j}. These are distributions that remain stable (i.e., players do not go too far from it) for a super-polynomial number of steps (rather than forever, as for logit equilibria). The hope is that these distributions exist and can be reached quickly by logit dynamics.

We identify a class of potential games, called asymptotically well-behaved, for which the behavior of the logit dynamics is not chaotic as the number of players increases so to guarantee meaningful asymptotic results. We prove that any such game admits distributions which are metastable no matter the level of noise present in the system, and the starting profile of the dynamics. These distributions can be quickly reached if the rationality level is not too big when compared to the inverse of the maximum difference in potential. Our proofs build on results which may be of independent interest, including some spectral characterizations of the transition matrix defined by logit dynamics for generic games and the relationship of several convergence measures for Markov chains.
\end{abstract}

\newpage

\section{Introduction}
One of the most prominent assumptions in game theory dictates that people are rational. This is contrasted by many concrete instances of people making irrational choices in certain strategic situations, such as stock markets \cite{bookIrrationality}. This might be due to the incapacity of exactly determining one's own utilities: the strategic game is played with utilities perturbed by some noise.

Logit dynamics \cite{blumeGEB93} incorporates this noise in players' actions and then is advocated to be a good model for people behavior. More in detail, logit dynamics features a rationality level $\beta \geq 0$ (equivalently, a noise level $1/\beta$) and each player is assumed to play a strategy with a probability which is proportional to her corresponding utility and $\beta$. So the higher $\beta$ is, the less noise there is and the more rational players are. Logit dynamics can then be seen as a noisy best-response dynamics.

The natural equilibrium concept for logit dynamics is defined by a probability distribution over the pure strategy profiles of the game. Whilst for best-response dynamics pure Nash equilibria are stable states, in logit dynamics there is a chance, which is inversely proportional to $\beta$, that players deviate from such strategy profiles. Pure Nash equilibria are then not an adequate solution concept for this dynamics. However, the random process defined by the logit dynamics can be modeled via an ergodic Markov chain. Stability in Markov chains is represented by the concept of stationary distributions. These distributions, dubbed logit equilibria, are suggested as a suitable solution concept in this context due to their properties \cite{afppSAGT10}. For example, from the results known in Markov chain literature, we know that any game possesses a logit equilibrium and that this equilibrium is unique. The absence of either of these guarantees is often considered a weakness of pure Nash equilibria. Nevertheless, as for Nash equilibria, the computation of logit equilibria may be computationally hard depending on whether the chain mixes rapidly or not \cite{afpppSPAA11j}.

As the hardness of computing Nash equilibria justifies approximate notions of the concept \cite{lmmEC03,csSODA07}, so Auletta et al. \cite{afppSODA12j} look at an approximation of logit equilibria that they call \emph{metastable distributions}.
These distributions aim to \emph{describe} regularities arising during the transient phase of the dynamics before stationarity has been reached.
Indeed, they are distributions that remain stable for a time which is long enough for the observer (in computer science terms, this time is assumed to be super-polynomial) rather than forever.
Roughly speaking, the stability of the distributions in this concept is measured in terms of the generations living some historical era, while stationary distributions remain stable throughout all the generations. When the convergence to logit equilibria is too slow, then there are generations which are outlived by the computation of the stationary distribution. For these generations, metastable distributions grant an otherwise impossible descriptive power. (We refer the interested reader to \cite{afppSODA12j} for a complete overview of the rationale of metastability.) It is unclear whether and which strategic games possess these distributions and if logit dynamics quickly reaches them.

The focus of this paper is the study of metastable distributions for the logit dynamics run on the class of potential games \cite{MS96}. Potential games are an important and widely studied class of games modeling many strategic settings. Each such game satisfies a number of appealing properties, the existence of pure Nash equilibria being one of them. A general study of metastability of potential games was left open by \cite{afppSODA12j} and assumes particular interest due to the known hardness results, see e.g. \cite{FPT04}, which suggest that the computation of pure Nash equilibria for them is an intractable problem, even for centralized algorithms. 

\paragraph{Our contribution.} 
We aim to prove asymptotic results, in terms of the number of players $n$ of potential games, concerning the super-polynomially long stability of metastable distributions, and the polynomial convergence time to them. This desiderata imposes some requirement on the potential games of interest, for otherwise a chaotic behavior (w.r.t. $n$) of the logit dynamics run on a game would not allow any meaningful asymptotic guarantee for it. We therefore identify a simple-to-describe class of potential games, termed \emph{asymptotically well-behaved}, for which the behavior of the dynamics is ``almost'' the same for any number of players. Intuitively, the potential function of a game in this class has a shape which is, in a sense, immaterial from the actual value of $n$. For example, the potential might be minimized when all the players agree on either strategy $x$ or $y$, maximized when half of the players play $x$ and the other half $y$, and increase as the number of players playing a strategy different from that played by the majority of players increases. The technical definition of this notion can be found in Section~\ref{sec:awd}. We stress that similar assumptions are made in related literature on logit dynamics either implicitly (as in \cite{msFOCS09,afpppSPAA11j}, where it is assumed that certain properties of the potential function do not change as $n$ changes), or explicitly, by considering specific games that clearly enjoy this property \cite{afppSODA12j}. Moreover, asymptotic results on the mixing time of Markov chains do require some assumption on the behavior of the chain (technically, the minimum bottleneck ratio must either be a polynomial or a super-polynomial) usually implicitly guaranteed by the definition of the chain at hand. Given that our objective is much more complex than bounding the mixing time (i.e., measuring asymptotically the transient phase of the chain -- defined on \emph{a} potential game -- and ascertain stability of \emph{and} convergence time to metastable distributions) a similar, yet stronger, requirement ought to be used. 


Together with the formalization of the class of games of interest, we formalize, building upon \cite{afppSODA12j}, the concept of asymptotic convergence/closeness to a metastable distribution, as a function of the number of players of a game. We then note, via the careful construction of an ad-hoc $n$-player potential game 
that not all potential games admit metastable distributions (cf. Section~\ref{sec:asymmeta}), thus showing formally that some form of restriction of games under consideration is necessary.

Our main result proves that any asymptotically well-behaved $n$-player potential game has a met\-a\-sta\-ble distribution for each starting profile of the logit dynamics.
These distributions remain stable for a time which is super-polynomial in $n$,
if one is content of being within distance $\varepsilon > 0$ from the distributions.
(The distance is defined in this context as the total variation distance, see below.) 
We also prove that the convergence rate to these distributions, called \emph{pseudo-mixing time},
is polynomial in $n$ for values of $\beta$ not too big
when compared to the (inverse of the) maximum difference in potential of neighboring profiles.
Note that when $\beta$ is very high then logit dynamics is ``close'' to the best-response dynamics
and therefore it is impossible to prove in general quick convergence results for potential games
due to the aforementioned hardness results.
We then give a picture which is as complete as possible relatively to the class of well-behaved potential games.
(To maintain $n$ as our only parameter of interest,
we assume that the logarithm of the number of strategies available to players is upper bounded by a polynomial in $n$;
this assumption can, however, be relaxed to prove bounds asymptotic in $n$
and in the logarithm of the maximum number of strategies.)

The proof of the above results consists of two main steps.
We first devise a sufficient property for any $n$-player (not necessarily potential) game 
to have, for any starting profile, a distribution that is metastable for a super-polynomial number of steps
and reached in polynomial time.
The main idea behind this sufficient condition
is that when the dynamics starts from a subset from which it is ``hard to leave'' and in which it is ``easy to mix'',
then the dynamics will stay for a long time close to the stationary distribution restricted to that subset.
Moreover, if a subset is ``easy-to-leave,'' then the dynamics will quickly reach a ``hard-to-leave'' subset.
The sufficient property  consists of a rather technical definition
that is intuitively a partition of the profiles into subsets that are asymptotically ``hard-to-leave \& easy-to-mix'' or ``easy-to-leave''.

The second step amounts to showing that any asymptotically well-behaved potential game admits such a partition. The proof of this result builds on a number of involved technical contributions, some of which might be of independent interest.
They mainly concern Markov chains.
The concepts of interest are mixing time (how long the chain takes to mix), bottleneck ratio (intuitively, how hard it is for the stationary distribution to leave a subset of states), hitting time (how long the chain takes to hit a certain subset of states) and spectral properties of the transition matrix of Markov chains. 
In particular, we define a procedure which computes the required partition for these games. We iteratively identify in the set of pure strategy profiles the ``hard-to-leave'' subsets.
To prove that these subsets are ``easy-to-mix'',
we firstly relate the pseudo-mixing time to the mixing time of a certain family of restricted Markov chains.
We then prove that the mixing time of these chains is polynomial by using a spectral characterization of the transition matrix of restricted Markov chains.
Finally, the proof that the remaining profiles are ``easy-to-leave'' mainly relies on a connection between bottleneck ratio and hitting time.
Specifically, we prove both an upper bound and a lower bound on the hitting time of a subset of states in terms of the bottleneck ratio of its complement.

We remark that, as a byproduct of our result, we essentially close an open problem of \cite{afppSODA12j} about metastability of the Curie-Weiss game.

In appendix, we complement the above contributions along two different dimensions. Firstly, the ``simple'' definition of asymptotic well-behaved games comes at the cost of sacrificing the full generality of the argument, in that a refinement of the condition actually suffices for our main result (see Section \ref{sec:awc}). Secondly, we prove further spectral results about the transition matrix of Markov chains defined by logit dynamics for a strategic (not necessarily potential) game  (cf. Section~\ref{sec:spectral}).
These results enhance our understanding of the dynamics and pave the way to further advancements in the area.

\hyphenation{Blu-me}
\paragraph{Related works.}
Blume \cite{blumeGEB93} introduced logit dynamics for modeling a noisy-rational behavior in game dynamics. Early works about this dynamics have focused on its long-term behavior: Blume \cite{blumeGEB93} showed that, for $2 \times 2$ coordination games
and potential games, the long-term behavior of the system is concentrated around a specific Nash
equilibrium; Al\`os-Ferrer and Netzer \cite{Ferrer2010} gave a general characterization of long-term behavior of logit dynamics for wider classes of games. Several works gave bounds on the time that the dynamics takes to reach specific Nash equilibria of a game: Ellison \cite{ellisonECO93} considered graphical coordination games on cliques and rings; Peyton Young \cite{youngTR00} and Montanari and Saberi \cite{msFOCS09} extended this work to more general families of graphs; Asadpour and Saberi \cite{asWINE09} focused on a class of congestion games.
Auletta et al. \cite{afppSAGT10} were the first to propose the stationary distribution of the logit dynamics Markov chain as a new equilibrium concept in game theory and to focus on the time the dynamics takes to get close to this equilibrium \cite{afpppSPAA11j}.

In physics, chemistry, and biology, metastability is a phenomenon related to the evolution of systems under noisy dynamics.
In particular, metastability concerns moves between regions of the state spaces and the existence of multiple, well separated time scales: at short time scales, the system appears to be in a quasi-equilibrium, but really explores only a confined region of the available space state, while, at larger time scales, it undergoes transitions between such different
regions.
Previous research about metastability aims at expressing typical features of a metastable state and to evaluate the transition time between metastable states. Several monographs on the subject are available in literature (see, for example, \cite{holSPA04,ovCUP05,bovICM06,holLNM09}). Auletta et al. \cite{afppSODA12j} applied metastability to probability distributions, introducing the concepts of metastable distribution and pseudo-mixing time and proving results for some specific potential games.

Roughly speaking, metastability is a kind of approximation for stationarity. From this point of view, metastable distributions may be likened to approximate equilibria. Two different approaches to approximated equilibria have been proposed in literature. In the multiplicative version \cite{csSODA07} a profile is an approximate equilibrium as long as each player gains at least a factor $(1 - \varepsilon)$ of the payoff she gets by playing any other strategy: these equilibria have been shown to be computationally hard both in general \cite{dasSODA11} and for congestion games \cite{svSTOC08}. In the additive version \cite{kmUAI02}, a profile is an approximate equilibrium as long as each player gains at least the payoff she gains by playing any other strategy minus a small additive factor $\varepsilon>0$: for these equilibria a quasi-polynomial time approximation scheme exists \cite{lmmEC03} but it is impossible to have an FPTAS \cite{cdtFOCS06}.

\section{Preliminary definitions}
A \emph{strategic game} $\G$ is a triple $([n], {\cal S}, {\cal U})$, where $[n] = \{1, \ldots, n\}$
is a finite set of players, ${\cal S} = (S_1, \ldots, S_n)$ is a family of non-empty finite sets ($S_i$ is the set of strategies available to player $i$), and ${\cal U} = (u_1, \ldots, u_n)$ is a family of utility functions (or payoffs), where $u_i \colon S \rightarrow \mathbb{R}$, $S = S_1 \times \ldots \times S_n$ being the set of all strategy profiles, is the utility function of player $i$. We let $m$ denote an upper bound to the size of players' strategy sets, that is, $m \geq \max_{i=1,\ldots,n} |S_i|$. We focus on (exact) \emph{potential games}, i.e., games for which there exists a function $\Phi \colon S \rightarrow \mathbb{R}$ such that for any pair of $\x, \y \in S$, $\y=(\x_{-i},y_i)$, we have:
$$
\Phi(\x) - \Phi(\y) = u_i(\y) - u_i(\x).
$$
Note that we use
the standard game theoretic notation $(\x_{-i},s)$ to mean the vector obtained from $\x$ by replacing the $i$-th entry with $s$;
i.e. $(\x_{-i},s)=(x_1,\ldots, x_{\scriptscriptstyle i-1},s,x_{\scriptscriptstyle i+1},\ldots,x_n)$.
A strategy profile $\x$ is a Nash equilibrium\footnote{In this paper, we only focus on pure Nash equilibria. We avoid explicitly mentioning it throughout.} if, for all $i$, $u_i(\x) \geq u_i(\x_{-i},s_i)$, for all $s_i \in S_i$. It is fairly easy to see that local minima of the potential function correspond to the Nash equilibria of the game.

For two vectors $\x,\y$, we denote with $H(\x,\y) = |\{i \colon x_i \neq y_i\}|$ the Hamming distance between $\x$ and $\y$. For every $\x \in S$, $N(\x) = \{\y \in S \colon H(\x,\y) = 1\}$ denotes the set of neighbors of $\x$ and $N_i(\x) = \{\y \in N(\x) \colon \y_{-i} = \x_{-i}\}$ is the set of those neighbors that differ exactly in the $i$-th coordinate.

In this paper, given a set of profiles $L$ we let $\overline{L}$ denote its complementary set, i.e., $\overline{L}= S \setminus L$.
Moreover, we say that a set $L$ is \emph{connected} if, for every $\x, \y \in L$, there are $\x_0, \x_1, \ldots, \x_k \in L$ with $\x_0 = \x$,
$\x_k = \y$ and $H(\x_{i-1},\x_i) = 1$ for each $i = 1, ldots, k$.

\subsection{Logit dynamics}
The logit dynamics has been introduced in~\cite{blumeGEB93} and runs as follows: at every time step (i) Select one player $i \in [n]$ uniformly at random; (ii) Update the strategy of player $i$ according to the \emph{Boltzmann distribution} with parameter $\beta$ over the set $S_i$ of her strategies. That is, a strategy $s_i \in S_i$ will be selected with probability
\begin{equation}\label{eq:updateprob}
\sigma_i(s_i \mid \x_{-i}) = \frac{1}{Z_i(\x_{-i})}  e^{\beta u_i(\x_{-i}, s_i)},
\end{equation}
where $\x_{-i}$ denotes the profile of strategies played at the current time step by players different from $i$, $Z_i(\x_{-i}) = \sum_{z_i \in S_i} e^{\beta u_i(\x_{-i}, z_i)}$ is the normalizing factor, and $\beta \geq 0$. One can see parameter $\beta$ as the inverse of the noise or, equivalently, the \emph{rationality level} of the system: indeed, from \eqref{eq:updateprob}, it is easy to see that for $\beta = 0$ player $i$ selects her strategy uniformly at random, for $\beta > 0$ the probability is biased toward strategies promising higher payoffs, and for $\beta$ that goes to infinity player $i$ chooses her best response strategy (if more than one best response is available, she chooses one of them uniformly at random).

The above dynamics defines a \emph{Markov chain} $\{ X_t \}_{t \in \mathbb{N}}$ with the set of  strategy profiles as state space, and where the transition probability from profile $\x = (x_1, \ldots, x_n)$ to profile $\y = (y_1, \ldots, y_n)$, denoted $P(\x,\y) = \Prob{\x}{X_1 = \y}$\footnote{Throughout this work, we denote with $\Prob{\x}{\cdot}$ the probability of an event conditioned on the starting state of the logit dynamics being $\x$.}, is zero if $H(\x,\y) \geq 2$ and it is $\frac{1}{n} \sigma_i(y_i \mid \x_{-i})$ if the two profiles differ exactly at player $i$. More formally, we can define the logit dynamics as follows.
\begin{definition}\label{def:LD}
Let $\G = ([n], \mathcal{S}, \mathcal{U})$ be a strategic game and let $\beta \geq 0$. The \emph{logit dynamics} for $\G$ is the Markov chain $\mathcal{M}_\beta = \left(\{ X_t \}_{t \in \mathbb{N}},S,P\right)$ where $S = S_1 \times \cdots \times S_n$ and
\begin{equation}\label{eq:transmatrix}
P(\x, \y) = \frac{1}{n} \cdot
\begin{cases}
\sigma_i(y_i \mid \x_{-i}), & \text{if } \y_{-i} = \x_{-i} \text{ and } y_i \neq x_i; \\
\overset{n}{\underset{i=1}{\sum}} \sigma_i(y_i \mid \x_{-i}), & \text{if } \y = \x; \\
0, & \text{otherwise;}
\end{cases}
\end{equation}
where $\sigma_i(y_i \mid \x_{-i})$ is defined in \eqref{eq:updateprob}.
\end{definition}
The Markov chain defined by \eqref{eq:transmatrix}
is ergodic~\cite{blumeGEB93}. Hence, from every initial profile $\x$ the distribution $P^t(\x, \cdot)$ over states of $S$ of the chain $X_t$ starting at $\x$ will eventually converge to a \emph{stationary distribution} $\pi$ as $t$ tends to infinity. As in \cite{afppSAGT10}, we call the stationary distribution $\pi$ of the Markov chain defined by the logit dynamics on a game $\G$, the \emph{logit equilibrium} of $\G$.
In general, a Markov chain with transition matrix $P$ and state space $S$ is said to be \emph{reversible} with respect to a distribution $\pi$ if, for all $\x,\y\in S$, it holds that
$
 \pi(\x) P(\x,\y)=\pi(\y) P(\y,\x).
$
If an ergodic chain is reversible with respect to $\pi$, then $\pi$ is its stationary distribution. Therefore when this happens, to simplify our exposition we simply say that the matrix $P$ is reversible. For the class of {potential games} the stationary distribution is the well-known \emph{Gibbs measure}.
\begin{theorem}[\cite{blumeGEB93}]
\label{thm:gibbsPot}
If $\G = ([n], \mathcal{S}, \mathcal{U})$ is a potential game with potential function $\Phi$, then the Markov chain given by \eqref{eq:transmatrix} is reversible with respect to the Gibbs measure
$
\pi(\x) = \frac{1}{Z} e^{-\beta \Phi(\x)},
$
where
$Z = \sum_{\y \in S} e^{-\beta \Phi(\y)}$
is the normalizing constant.
\end{theorem}
It is worthwhile to notice that logit dynamics for potential games and Glauber dynamics for Gibbs distributions are two ways of looking at the same Markov chain (see \cite{blumeGEB93} for details). This, in particular, implies that we can write
$$\sigma_i(s_i \mid \x_{-i}) = \frac{e^{-\beta \Phi(\x_{-i}, s_i)}}{\sum_{z \in S_i} e^{-\beta \Phi(\x_{-i}, z)}}.$$

\subsection{Convergence of Markov chains}
\paragraph{Mixing time.}
Arguably, the principal notion to measure the rate of convergence of a Markov chain to its stationary distribution is the \emph{mixing time}, which is defined as follows. Let us set
$$
d(t)=
\max_{\x\in S}
\tv{P^t(\x,\cdot) - \pi},
$$
where the {\em total variation distance} $\tv{\mu - \nu}$ between two probability distributions $\mu$ and $\nu$ on the same state space $S$ is  defined as
$$
\tv{\mu - \nu}=\max_{A\subset S}
	|\mu(A)-\nu(A)| = \frac{1}{2} \sum_{\x \in S} | \mu(\x) - \nu(\x) |.
$$
For $0 < \varepsilon < 1/2$, the mixing time of the logit dynamics is  defined as
$$
\tm(\varepsilon) =
\min \{t\in\mathbb{N} \colon d(t)\leq\varepsilon\}.
$$
It is usual to set $\varepsilon = 1/4$ or $\varepsilon = 1/2e$. We write $\tm$ to mean $\tm(1/4)$ and we refer generically to ``mixing time'' when the actual value of $\varepsilon$ is immaterial. Observe that $\tm(\varepsilon)\leq\left\lceil\log_2 \varepsilon^{-1}\right\rceil\tm$.

\paragraph{Relaxation time.}
Another important measure of convergence for Markov chains is given by the \emph{relaxation time}. Let $P$ be the transition matrix of a Markov chain with finite state space $S$; let us label the eigenvalues of $P$ in non-increasing order
$$
\lambda_1\geq \lambda_2 \geq \dots \geq \lambda_{|S|}.
$$
It is well-known (see, for example, Lemma~12.1 in~\cite{lpwAMS08}) that $\lambda_1 = 1$ and, if $P$ is ergodic, then $\lambda_2<1$ and $\lambda_{|S|}>-1$. We set $\lambda^\star$ as the largest eigenvalue in absolute value other than $\lambda_1$,
$$
\lambda^\star = \max_{i=2, \ldots, |S|} \left\{ |\lambda_i| \right\}.
$$
The relaxation time $\tr$ of a Markov chain $\mathcal{M}$ is defined as
$$
\tr = {\frac{1}{1-\lambda^\star}}.
$$
The relaxation time is related to the mixing time by the following theorem (see, for example, Theorems 12.3 and 12.4 in \cite{lpwAMS08}).
\begin{theorem}[Relaxation time]\label{theorem:relaxation}
Let $P$ be the transition matrix of an ergodic and reversible Markov chain with state space
$S$ and stationary distribution $\pi$. Then
$$
(\tr-1)\log 2
\leq \tm\leq
\log\left({\frac{4}{\pimin}}\right) \tr,
$$
where
$\pimin=\min_{\x \in S} \pi(\x)$.
\end{theorem}

\paragraph{Hitting time.}
In some cases, we are interested in bounding the first time that the chain hits a profile in a certain set of states, also known as its \emph{hitting time}. Formally, for a set $L \subseteq S$, we denote by $\tauset{L}$ the random variable denoting the hitting time of $L$. Note that the hitting time, differently from mixing
and relaxation
time, depends on where the dynamics starts.
Some useful facts about hitting time are summarized in Appendix~\ref{sec:hitting}.

\paragraph{Bottleneck ratio.}
Quite central in our study is the concept of \emph{bottleneck ratio}. Consider an ergodic Markov chain with finite state space $S$, transition matrix $P$, and stationary distribution $\pi$. The probability distribution $Q(\x, \y) = \pi(\x)P(\x,\y)$  is of particular interest and is sometimes called the \emph{edge stationary distribution}. Note that if the chain is reversible then $Q(\x, \y)= Q(\y, \x)$. For any $L \subseteq S$, $L \neq \emptyset$, we let $Q(L, S 
\setminus L)=\sum_{\x \in L, \y \in S \setminus L} Q(\x,\y)$. Then the bottleneck ratio of $L$ is
$$
B(L) = \frac{Q(L,S \setminus L)}{\pi(L)}.
$$
We use the following theorem to derive lower bounds to the mixing time (see, for example, Theorem~7.3 in \cite{lpwAMS08}).
\begin{theorem}[Bottleneck ratio]\label{theorem:bottleneck}
Let $\mathcal{M} = \{ X_t \colon t \in \mathbb{N} \}$ be an ergodic Markov
chain with state space $S$,
transition matrix $P$, and stationary distribution $\pi$.
Let $L \subseteq S$ be any set with $\pi(L) \leq 1/2$.
Then the mixing time is
$$
\tm \geq \frac{1}{4 B(L)}.
$$
\end{theorem}
The bottleneck ratio is also strictly related to the relaxation time. Indeed, let
$$
B_\star = \min_{L \colon \pi(R) \leq 1/2} B(L),
$$
then the following theorem holds (see, for example, Theorem~13.14 in \cite{lpwAMS08}).
\begin{theorem}
 \label{thm:bottle_rel}
 Let $P$ be the transition matrix of an ergodic and reversible Markov chain with state space $S$. Let $\lambda_2$ be the second largest eigenvalue of $P$. Then
$$
\frac{B_\star^2}{2}
\leq 1 - \lambda_2 \leq
2 B_\star.
$$
\end{theorem}

\section{Metastability}
In this section we give formal definitions of \emph{metastable distributions} and \emph{pseudo-mixing time}.
We also survey some of the tools used for our results.
For a more detailed description we refer the reader to~\cite{afppSODA12j}.

\begin{definition}\label{def:metastability}
Let $P$ be the transition matrix of a Markov chain with finite state space $S$. A probability distribution $\mu$ over $S$ is $(\varepsilon,\Time)$-\emph{metastable} for $P$ (or simply metastable, for short) if for every $0 \leq t \leq \Time$ it holds that
$$
\tv{\mu P^t - \mu} \leq \varepsilon.
$$
\end{definition}

The definition of metastable distribution captures the idea of a distribution that behaves approximately like the stationary distribution: if we start from such a distribution and run the chain we stay close to it for a ``long'' time. Some interesting properties of metastable distributions are discussed in \cite{afppSODA12j}, including the following lemmata, that turn out to be useful for proving our results.
\begin{lemma}[\cite{afppSODA12j}]
\label{lem:meta:bottleneck}
Let $P$ be a Markov chain with finite state space $S$ and stationary distribution $\pi$. For
a subset of states $L\subseteq S$ let $\pi_L$ be the stationary distribution conditioned on $L$, i.e.
\begin{equation}
\label{eq:piL}
\pi_L(\x)=\begin{cases}
\pi(\x)/\pi(L),& \text{if } \x\in L;\\
0,& \text{otherwise.}
\end{cases}
\end{equation}
Then, $\pi_L$ is $(B(L),1)$-metastable.
\end{lemma}

\begin{lemma}[\cite{afppSODA12j}]
\label{lem:meta:1}
If $\mu$ is $(\varepsilon,1)$-metastable for $P$ then $\mu$ is
$(\varepsilon \Time, \Time)$-metastable for $P$.
\end{lemma}

Among all metastable distributions, we are interested in the ones that are quickly reached from a (possibly large) set of states. This motivates the following definition.

\begin{definition}\label{def:pseudomixing}
Let $P$ be the transition matrix of a Markov chain with state space $S$, let $L \subseteq S$ be a non-empty set of states and let $\mu$ be a probability distribution over $S$. We define the \emph{pseudo-mixing time} $\pmt{\mu}{L}{\varepsilon}$ as
$$
\pmt{\mu}{L}{\varepsilon} = \inf \{ t \in \mathbb{N} \colon \tv{P^t(\x,\cdot) - \mu} \leq \varepsilon \mbox{ for all } \x \in L \}.
$$
\end{definition}
Since the stationary distribution $\pi$ of an ergodic Markov chain is reached within $\varepsilon$ in time $\tm(\varepsilon)$ from every state, according to Definition~\ref{def:pseudomixing} we have that $\pmt{\pi}{S}{\varepsilon} = \tm(\varepsilon)$.
The following simple lemma connects metastability and pseudo-mixing time.
\begin{lemma}[\cite{afppSODA12j}]\label{lem:metaonetot}
Let $\mu$ be a $(\varepsilon,\Time)$-metastable distribution and let $L \subseteq S$ be a set of states such that $t_\mu^L(\varepsilon)$ is finite. Then for every $\x \in L$ it holds that
$
\tv{P^t(\x,\cdot) - \mu} \leq 2 \varepsilon \mbox{ for every } t_\mu^L(\varepsilon) \leq t \leq t_\mu^L(\varepsilon) + \Time.
$
\end{lemma}

\subsection{Asymptotic metastability}
\label{sec:asymmeta}
The notions and results introduced above apply to a single Markov chain.
Auletta et al. \cite{afppSODA12j} adopted these notions to evaluate the behavior
of the logit dynamics for potential games, as $n$ grows. 
Therefore, we do not have a single Markov chain
but a sequence of them, one for each number $n$ of players,
and need to consider an asymptotic counterpart of the notions above.
Auletta et al. \cite{afppSODA12j}, in fact, showed that the logit dynamics for specific classes of $n$-player potential games enjoys the following property,
that we name \emph{asymptotic metastability}.
\begin{definition}\label{def:asmetastability}
Let $\G$ be an $n$-player strategic game.
We say that the logit dynamics for $\G$ is \emph{asymptotically metastable} for the rationality level $\beta$
if there are constants $n_0, \varepsilon > 0$, a polynomial $p=p_\varepsilon$ and a super-polynomial $q=q_\varepsilon$
such that for each $n \geq n_0$, the logit dynamics for the $n$-player game $\G$
converges in time at most $p(n)$ from each profile of $\G$ to a $(\varepsilon, q(n))$-metastable distribution.
\end{definition}
When the logit dynamics for a game is (not) asymptotically metastable, we say for brevity that the game itself is (not) asymptotically metastable. Unfortunately, asymptotic metastability cannot be proved for every (potential) game as shown next.
\begin{lemma}
\label{lem:potg_no_asymMeta}
There is a $n$-player (potential) game $\G$ 
which 
is not asymptotically metastable for any $\beta$ sufficiently high and any $\varepsilon < \frac{1}{4}$. 
\end{lemma}
\begin{proof}
We will show a game $\G$ and a profile $\x$ of this game such that  for any $0 < \varepsilon < 1/4$, for infinitely many value of $n$ and for each polynomial $p$ in $n$ and each super-polynomial $q$ in $n$, the logit dynamics for $\G$ does not converge in time at most $p(n)$ from $\x$ to any $(\varepsilon, q(n))$-metastable distribution, even if the mixing time of the dynamics is larger than $p$.

Consider now the following pairs $(p_j, q_j)$, where $p_j = n^j$ and $q_j = \exp\left(\log n \cdot {\log^{(j)}} n\right)$,
where $\log^{(j)}$ is the $j$-th functional iteration of the logarithm function.
Let us denote as $n_j$ a value such that $p_j(n_j) < q_j(n_j) - \varepsilon$.
Such a value surely exists since $p$ is polynomial and $q$ is super-polynomial.
Moreover, observe that for each $n > n_j$, we have $p_j(n) < q_j(n) - \varepsilon$.
Thus, we can assume without loss of generality that $1 = n_0 < n_1 < n_2 < \ldots$.
Now let $\Time$ be a function that is asymptotically sandwiched between $p_j$ and $q_j$, for any $j$.
This can be guaranteed by letting $\Time$ be a function such that
$\Time(n) = q_j(n) - \varepsilon$ for $j$ such that $n_{j-1} < n \leq n_j$.
Note that for any $p_j$ and for any $n \geq n_j$, we have
$\Time(n) = q_k(n) - \varepsilon > p_k(n) \geq p_j(n)$,
where $k \geq j$ is such that $n_{k-1} < n \leq n_k$.
Similarly, for any $q_j$ and for any $n \geq n_j$ we have
$\Time(n) = q_k(n) - \varepsilon < q_k(n) \leq q_j(n)$. The situation is depicted in Figure \ref{fig:AWDpotcl}.

\begin{figure}[ht]
\centering
\includegraphics{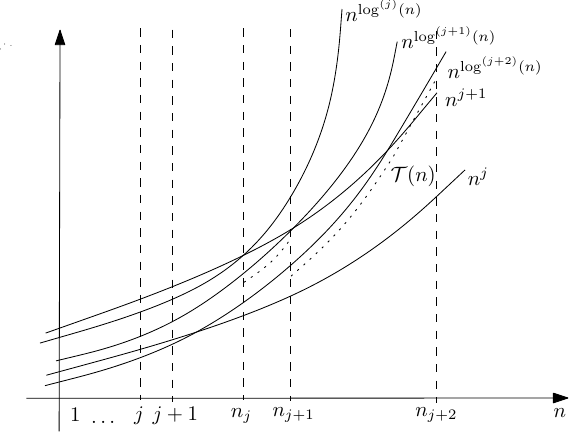}
\caption{The figure shows how $\Time(n)$ is built around the functions $p_j$'s and $q_j$'s
so that it is hard to classify $\Time(n)$ either as a polynomial or as a super-polynomial.
Note that $\Time(n)$  is the inverse of the bottleneck ratio of profile $(1, \ldots, 1)$
and thus it describes the time needed to leave that profile. }\label{fig:AWDpotcl}
\end{figure}

Let now $\G$ be a $n$-player potential game such that each player has exactly two strategies, say $0$ and $1$.
Consider the potential function $\Phi$ such that
for any $t = 0, \ldots, n-1$ and any profile $\x$ wherein exactly $t$ players play strategy $1$ we have 
$\Phi(\x) = n - t$, while $\Phi(1, \ldots, 1) = 1 + k_n$,
where $k_n = \frac{1}{\beta} \log \left(\frac{\Time(n)}{\varepsilon} - 1\right)$,
$\beta$ being the rationality parameter of the logit dynamics. 
Observe that if there is a pair $(p, q)$ with $p$ polynomial in $n$ and $q$ super-polynomial in $n$
such that it is possible to prove that the logit dynamics for $\G$ is asymptotic metastable with parameters $p$ and $q$,
then there is $j^\star$ such that the results holds also with $(p_{j^\star}, q_{j^\star})$ in place of $p$ and $q$,
where $(p_{j^\star}, q_{j^\star})$ corresponds to one of the pair of  functions described above in the definition of the game $\G$.
Hence, in order to prove the lemma is sufficient to show that it holds only for pairs $(p_j,q_j)$ as described above.

Note that, by taking $\beta$ sufficiently high, we have that: (i) $\pi(0, \ldots, 0) \geq \frac{1}{2}$; (ii) there exists a $j$ such that, for any subset
$L \subseteq \{0,1\}^n \setminus \{(0, \ldots, 0), (1, \ldots, 1)\}$, the bottleneck ratio $B(L)$ is at least the inverse of $p_j$; (iii) the bottleneck ratio $B(1, \ldots, 1) = \frac{\varepsilon}{\Time(n)}$. 

Firstly note that the mixing time of logit dynamics for $\G$ is not polynomial. Indeed, from Theorem~\ref{theorem:bottleneck}, it follows that the mixing time is at least $\frac{\Time}{4\varepsilon}$.
However, as suggested above, for each $p_j$, we have $\frac{\Time}{4\varepsilon} > \Time > p_j$ for infinitely many $n$, and hence the mixing time is asymptotically greater than any polynomial $p_j$.

We next discuss that no metastable distribution is stable for a super-polynomial time or, even if there is one, it cannot be reached in polynomial time. From Lemma~\ref{lem:meta:bottleneck} and Lemma~\ref{lem:meta:1}, we have that for each $n$ the distribution $\pi_1$ that assigns probability 1 to the profile $(1, \ldots, 1)$ is $(\varepsilon, \Time(n))$-metastable. However, as suggested above, for each $q_j$, the function $\Time$ is smaller than $q_j$ for infinitely many $n$. Thus, the distribution $\pi_1$ is metastable for time that is asymptotically smaller than any super-polynomial $q_j$.
Note that this argument extends to any $(3 \varepsilon, \Time(n))$-metastable distribution $\mu$ that is within distance $2\varepsilon$ from $\pi_1$.
Finally, observe that, the remaining distributions that are far from $\pi_1$ cannot be reached quickly from $(1,\ldots, 1)$.
In fact, from Lemma~\ref{lemma:hitting_bottle_ub} below, for any polynomial $p_j$ the probability that the logit dynamics leaves the profile $(1, \ldots, 1)$
in $p_j$ steps, is at most $\frac{\varepsilon \cdot p_j}{\Time - \varepsilon} < \varepsilon$, for $n$ sufficiently large.
Hence, for any $p_j$, starting from $(1, \ldots, 1)$ the pseudo-mixing time of any distribution $\mu$ that is at least $2\varepsilon$-far
from $\pi_1$ is asymptotically greater than $p_j$.
\end{proof}

\begin{remark}
\label{remark:weak_notion}
 The game described in the proof of Lemma~\ref{lem:potg_no_asymMeta} 
 also shows the necessity of having a definition of asymptotic metastability as the one given in Definition \ref{def:asmetastability}.
 
 Consider, indeed, the weaker definition of asymptotic metastability in which for each $n$ there is a polynomial $p_n(n)$ and a super-polynomial $q_n(n)$ governing convergence and stability time of metastability, respectively (i.e., a definition in which the order of quantifiers is reversed).
 This concept might at first glance look meaningful.
 However, it is instead of scarce significance as the distinction between polynomials and super-polynomials might become null in the limit.
 
 The game in Lemma~\ref{lem:potg_no_asymMeta} exemplifies this phenomenon,
 since it does not satisfies the metastability notion given in the one in Definition \ref{def:asmetastability},
 but it satisfies this weak notion.
 Indeed, for any $n$, there is a super-polynomial function, namely $\tilde{q}_n = q_j - \varepsilon$ for $j$ such that $n_{j-1} < n \leq n_j$,
 such that $\pi_1$ is $(\varepsilon, \tilde{q}_n(n))$-metastable. Obviously, the pseudo-mixing time of this distribution from the profile $(1, \ldots, 1)$
 is 1. From the remaining profiles, the dynamics quickly converges to the stationary distribution for any $\beta$ sufficiently large
 (this follows from well-known results about birth-and-death chains).
\end{remark}

Motivated by the result above, we next give a sufficient property for \emph{any} (not necessarily potential) game to be asymptotically metastable. Such a condition will help identifying the class of asymptotically metastable potential games.

\section{Asymptotic metastability and partitioned games}
\label{sec:partition}
In this section we will introduce the concept of game \emph{partitioned} by the logit dynamics. Then we give examples of games satisfying this notion. Finally, we prove that games partitioned by the logit dynamics are asymptotically metastable.

Henceforth, we will assume that the logarithm of the maximal number of strategies available to a player is at most a polynomial in $n$. Specifically, we denote as $m(\cdot)$ the function such that $m(n)$ is the maximum number of strategies available to a player in $\G$ when the number of players is $n$. Then, we will assume that the function $\log m(\cdot)$ is at most polynomial in its input. We can easily drop this assumption by asking for results that are asymptotic in $\log |S|$, where $|S|$ denotes the function returning the number of profiles of the game: each one of our proof can be rewritten according to this requirement with very small changes. Note that having results asymptotic in the logarithm of the number of states is a common requirement in Markov chain literature. Moreover, since $|S|$ for a game with $n$ players is at most $m(n)^n$, this requirement is equivalent to asking for results asymptotic in $n$ and in the logarithm of the function $m$.

Note also that we focus only on $n$-player (not necessarily potential) games and values of $\beta$ such that 
the mixing time of the logit dynamics 
is at least super-polynomial in $n$,
otherwise the stationary distribution enjoys the desired properties of stability and convergence.
Throughout the rest of the paper we will denote with $\beta_0$ the smaller value of $\beta$ such that the mixing time is not polynomial.

\subsection{Games partitioned by the logit dynamics}
Let $\G$ be an $n$-player game.
Let $P$ be the transition matrix of the logit dynamics on $\G$ and let $\pi$ be the corresponding stationary distribution.
For $L \subseteq S$ non-empty, we define a Markov chain with state space $L$ and transition matrix $\mathring{P}_L$ defined as follows.
\begin{equation}
\label{eq:restricted_loop}
 \mathring{P}_L(\x,\y) = \begin{cases}
                          P(\x, \y) & \text{if }
                          \x \neq \y;\\
                          1 - \sum_{\substack{\z \in L,\\\z \neq \x}}P(\x, \z) = P(\x,\x) + \sum_{\z \in S \setminus L}P(\x, \z) & \text{otherwise}.
                         \end{cases}
\end{equation}
It easy to check that the stationary distribution of this Markov chain is given by the distribution
$
\pi_L(\x) = \frac{\pi(\x)}{\pi(L)}
$,
for every $\x \in L$.
Note also that the Markov chain defined upon $\mathring{P}_L$ is aperiodic, since the Markov chain defined upon $P$ is,
and it will be irreducible if $L$ is a connected set.
For a fixed $\varepsilon > 0$, we will denote with $\tm^L(\varepsilon)$ the mixing time of the chain described in \eqref{eq:restricted_loop}.
We also denote with $B_{L}(A)$ the bottleneck ratio of $A \subset L$ in the Markov chain with state space $L$ and transition matrix $\mathring{P}_{L}$.

We are now ready to introduce the definition of partitioned games.
\begin{definition}
 Let $\G$ be an $n$-player strategic game and let us denote with $S$ its profile space.
 We say that $\G$ is \emph{partitioned} by the logit dynamics for the rationality level $\beta$
 if there are constants $n_0, \varepsilon > 0$,
 a polynomial $p=p_\varepsilon$ and a super-polynomial $q=q_\varepsilon$
 such that for each $n \geq n_0$
 there is a family of connected subsets $R_1, \ldots, R_k$ of $S$, with $k \geq 1$,
 and a partition $T_1, \ldots, T_k, N$ of $S$, with $T_i \subseteq R_i$ for any $i = 1, \ldots, k$, such that
 \begin{enumerate}
  \item the bottleneck ratio of $R_i$ is at most $1/q(n)$, for any $i = 1, \ldots, k$;
  \item the mixing time $\tm^{R_i}(\varepsilon)$ is at most $p(n)$, for any $i = 1, \ldots, k$;
  \item for any $i = 1, \ldots, k$ and for any $\x \in T_i$, it holds that
  $$\Prob{\x}{\tauset{S\setminus R_i} \leq \tm^{R_i}(\varepsilon)} \leq \varepsilon;$$
  \item for any $\x \in N$, it holds that
  $$\Prob{\x}{\tauset{\bigcup_i T_i} \leq p(n)} \geq 1 - \varepsilon.$$
 \end{enumerate}
\end{definition}
Note that we allow in the above definition that $T_i$, for some $i = 1, \ldots, k$, or $N$ are empty. In this context, we call $T_i$ the \emph{core} of $R_i$, $i \leq k$. 

The main result of this section proves that a game partitioned by the logit dynamics is asymptotically metastable. 
\begin{theorem}\label{thm:part_meta}
Let $\G$ be $n$-player game.
If $\G$ is partitioned by the logit dynamics for $\beta$,
then the logit dynamics for $\G$ is asymptotically metastable 
for the given $\beta$.
\end{theorem}

Before going into the details of the proof, we relate the technical notion of partitioned games with examples of actual games satisfying it as well as comment on games not enjoying it. All of our examples will be potential games.


\subsection{Examples of partitioned games}
\label{subsec:ex_part}
Three classes of $n$-player games,
namely pure coordination games,
Curie-Weiss games and graphical coordination games on the ring, have been proven asymptotically metastable \cite{afppSODA12j}.
Here we will show that any of these games is partitioned by the logit dynamics, whereas the game of Lemma \ref{lem:potg_no_asymMeta} is not partitioned by the logit dynamics.

\paragraph{Pure coordination games.}
The \emph{pure coordination game} is an $n$-player game where players have the same strategy set $A$ and each player is happy when
all players adopt the same strategy and unhappy otherwise.

Specifically, in \cite{afppSODA12j}, they consider the case in which each agent can choose between two strategies, namely $+1$ and $-1$; each agent has utility $1$ if all the players adopt the same strategy and utility $0$ otherwise. 
The mixing time of the logit dynamics for these games is polynomial for $\beta = \OO(\log n)$ and super-polynomial otherwise.
Auletta et al. \cite{afppSODA12j} show asymptotic metastability for any $\beta = \omega(\log n)$.
Next lemma proves that $n$-player pure coordination games are partitioned by the logit dynamics.
\begin{lemma}
 \label{lem:pure_coord_partition}
 Let $\G$ be an $n$-player pure coordination game. Then $\G$ is partitioned by the logit dynamics for any $\beta = \omega(\log n)$.
\end{lemma}
\begin{proof}
 Consider the following subsets of $S$: $R_1 = \{\p\}, R_2 = \{\m\}$ and $R_3 = \{+1,-1\}^n \setminus \{\p, \m\}$,
 where $\p = (+1)^n$ and $\m = (-1)^n$.
 As showed in \cite{afppSODA12j}, the bottleneck ratio of these subsets is super-polynomial for any $\beta = \omega(\log n)$.
 Moreover, the mixing time of the chains restricted to $R_1$ and $R_2$ is trivially polynomial.
 As for $R_3$, observe that the stationary distribution of the restricted chain
 is very close to the stationary distribution of a lazy random walk on a $n$-dimensional hypercube,
 whose mixing time is known to be polynomial (see, e.g., \cite{lpwAMS08}).
 
 Also let us consider the following partition of $S$: $T_i = R_i$ for $i = 1, 2, 3$ and $N$ is empty.
 Clearly, $T_1$ and $T_2$ satisfy the property required by the definition of partitioned games.
 This holds also for $T_3$. Indeed, consider a birth and death chain (see Section~2.5 in \cite{lpwAMS08})
 defined on the state space $\{0, 1, \ldots, m\}$
 with transition probability:
 $$
  p_0 = q_m = r_0 = r_m = \frac{1}{2}, q_0 = p_m = 0; \qquad
  p_i = \frac{m-i}{4m}, q_i = \frac{m+i}{4m}, r_i = \frac{1}{2}, \text{ for } i = 1, \ldots, m-1;
 $$
 where $p_i$ is the probability of going from state $i$ to state $i+1$,
 $q_i$ is the probability of going from state $i$ to state $i-1$
 and $r_i$ is the probability to stay in state $i$.
 
 Now observe that the expected hitting time of either $\p$ or $\m$ is equivalent to
 the expected hitting time of state\footnote{Here, we are assuming $n$ is even. The case for odd $n$ is similar.} $m = n/2$
 in the above birth an death chain. Indeed, the latter can be seen as the projection of our chain,
 where the state $\x$ of our chain is projected to the state $i$ of the birth and death chain,
 such that the minimum among the zeros and the ones in $\x$ is $\frac{n}{2}-i$.
 It is then easy to check that the expected hitting time of $m$ is super-polynomial in $n$ from any starting state
 (see Section~2.5 in \cite{lpwAMS08}). The claim finally follows by a simple application of Markov's inequality.
\end{proof}

\paragraph{The Curie-Weiss model.}
Consider now the following game-theoretic formulation of the well-studied \emph{Curie-Weiss model} (the \emph{Ising model} on the complete graph), that we will call \emph{\cw-game}: each one of $n$ players has two strategies, $-1$ and $+1$, and the utility of player $i$ for profile $\x = (x_1, \ldots, x_n) \in \{-1, +1\}^n$ is $u_i(\x) = x_i \sum_{j \neq i} x_j$.
Observe that for every player $i$ it holds that
$$u_i(\x_{-i}, +1) - u_i(\x_{-i}, -1) = {\cal H}(\x_{-i}, -1) - {\cal H}(\x_{-i}, +1),$$
where ${\cal H}(\x) = - \sum_{j\neq k} x_j x_k$, hence the \cw-game is a potential game with potential function $\cal H$.

It is known (see Chapter 15 in \cite{lpwAMS08}) that the logit dynamics for this game (or equivalently the \emph{Glauber dynamics} for the Curie-Weiss model) has mixing time polynomial in $n$ for $\beta < 1 / n$ and super-polynomial as long as $\beta > 1/n$. Moreover, \cite{afppSODA12j} describes metastable distributions for $\beta > c \log n /n$ and shows that such distributions are quickly reached from profiles where the number of $+1$ (respectively $-1$) is a sufficiently large majority, namely if the magnetization $k$ is such that $k^2 > c \log n/\beta$, where the \emph{magnetization} of a profile $\x$ is defined as $M(\x) = \sum_i x_i$.

It has been left open what happens when $\beta$ lies in the interval $(1/n , c \log n /n)$
and if a metastable distribution is quickly reached when in the starting point the number of $+1$ is close to the number of $-1$.
We observe that next lemma, along with Theorem~\ref{thm:part_meta}
essentially closes this problem by showing that \cw-games are asymptotic metastable for $\beta \geq c/n$ for some constant $c > 1$.

\begin{lemma}
\label{lem:curie_Weiss_partition}
 Let $\G$ be the $n$-player \cw-game. Then $\G$ is partitioned by the logit dynamics for any $\beta > c/n$, for constant $c > 1$.
\end{lemma}
\begin{proof}
 Let $S_+$ (resp., $S_-$) be the set of profiles with positive (resp., negative) magnetization
 and let us set $R_1 = S_+$ and $R_2 = S_-$.
 It is known that the bottleneck ratio of these subset is super-polynomial for any $\beta > c/n$, for constant $c > 1$, (see, e.g., Chapter~15 in \cite{lpwAMS08}). Moreover, in \cite{llp2010} it has been proved that the mixing time of the chain restricted to $S_+$ (resp. $S_-$)
 is actually $c_1 n \log n$  for some constant $c_1 > 0$%
 \footnote{The result in \cite{llp2010} refers to censored chains, that are exactly the same as 
 our restricted chain, except that the probability
 that the original chain from a profile $\x$ goes out from $L$ is ``reflected'' to some profile in $L$ different from $\x$,
 instead than being ``added'' to the probability to stay in $\x$.
 It is immediate to see how their result extends also to our restricted chains.}.
 
 Let now $\zeta$ be the unique positive root \cite{ding2009censored,ding2009mixing} of the function
 $$
  f(x) = \frac{e^{\beta x}(1- x) - e^{-\beta x}(1+x)}{e^{\beta x}(1- x) + e^{-\beta x}(1+x)}.
 $$
 Observe that $\zeta \in [0,1]$, it is non decreasing in $\beta$ and does not depend on $n$.
 Let now $Z_+$ be the set of profiles with magnetization $k \geq \zeta n$
 and $Z_-$ be the set of profiles with magnetization $k \leq -\zeta n$.
 Note that for a constant $c > 1$ and $n$ sufficiently large, we have that $|k| \geq 1$  \cite{ding2009mixing}.
 
 Consider now the following partition of $S$: $T_1 = Z_+$, $T_2 = Z_-$ and $N = S \setminus (Z_+ \cup Z_-)$.
 We prove that from any profile $\x \in Z_+$
 the dynamics hits a profile $\y \in S_-$ in a time equivalent to the mixing time of the chain restricted to $S_+$
 with probability at most $\varepsilon$.
 Consider, indeed, the magnetization chain, i.e., the birth and death chain on the space $\{-n, 2 - n, \ldots, n- 2, n\}$.
 Then we are interested in the hitting time $\tau_l$ of $l \leq 0$ when the starting point is $k$.
 Clearly, in order to reach magnetization $l$ it is necessary to reach magnetization $k'$, with $l < k' < k$.
 And for reaching $l$ from $k'$ it is necessary to reach $k''$ such that $l \leq k'' < k'$.
 Then, we show that there is $k'$ from which the chain quickly goes back to $n \zeta$ with high probability without ever hitting the profile $k''$.
 In particular, in \cite[Theorem 4.10]{ding2009censored} it has been showed that there are $k'$ and $k''$ such that
 $$
  \Prob{k'}{\tau_{n \zeta} \leq c_2 n \log n} \geq 1 - o(1) \qquad \text{and} \qquad \Prob{k'}{\tau_{k''} \geq c_2 n \log n} \geq 1 - o(1).
 $$
 Hence, it follows that
 $$
  \Prob{k}{\tau_l \geq c_1 n \log n} \geq \left(1 - o(1)\right)^\kappa = 1 - o(1),
 $$
 where $\kappa$ is a constant.
 (Clearly, everything holds symmetrically by considering the set $Z_-$.)

 Finally, observe that from \cite{ding2009censored} (Theorems~4.4, 4.9 and 4.10 -- see also \cite{ding2009mixing}),
 we have that for each profile $\x \in N$, the
 hitting time of $T_1 \cup T_2$ is polynomial with high probability.
\end{proof}

\paragraph{Graphical coordination game on the ring.}
Finally consider a \emph{graphical coordination game on the ring},
in which the $n$ players are identified by the vertices of a ring $R=(V, E)$.
For each edge of $R$, the endpoints play the following two-strategy coordination game:
$$
\begin{game}{2}{2}
      &$+1$      &$-1$ \\
  $+1$ &$a,a$    &$c,d$ \\
  $-1$ &$d,c$    &$b,b$
\end{game}\hspace*{\fill}
$$
where $a>d$ and $b>c$.
Each player picks a strategy and uses it for each coordination game in which she is involved.
The utility of a player is the sum of the utilities for each coordination game she plays.

It is known \cite{afpppSPAA11j,bkmp2005} that the mixing time of the logit dynamics for this game is polynomial in $n$ for $\beta = \OO(\log n)$
and greater than any polynomial in $n$ for $\beta = \omega(\log n)$. Moreover, in \cite{afppSODA12j} it has been showed that for this game asymptotic metastability holds. 
\begin{lemma}
 \label{lem:coord_ring_partition}
 Let $\G$ be an $n$-player graphical coordination game on the ring. Then $\G$ is partitioned by the logit dynamics for any $\beta = \omega(\log n)$.
\end{lemma}
\begin{proof}
 Let $R_1 = \{\p\}$ and $R_2 = \{\m\}$, $\p$ and $\m$ defined as in the proof of Lemma \ref{lem:pure_coord_partition}.
 Note that the bottleneck ratio of $R_1$ and $R_2$ is known to be super-polynomial for $\beta = \omega(\log n)$
 and the mixing time of the chains restricted to $R_1$ and $R_2$ is trivially polynomial.
 
 Consider now the following partition of $S$: $T_1 = \{\p\}$, $T_2 = \{\m\}$ and $N = \{+1, -1\}^n \setminus \{\p, \m\}$.
 Clearly, $T_1$ and $T_2$ satisfy the property required by the definition of partitioned games.
 Finally, the results in \cite{afppSODA12j} show that the hitting time of either $\p$ or $\m$
 when the chain starts from $\x \in N$ is polynomial.
\end{proof}

\paragraph{The game of Lemma \ref{lem:potg_no_asymMeta}.}
Finally, we consider the non-asymptotically metastable game of Lemma \ref{lem:potg_no_asymMeta} and prove the following corollary of Theorem \ref{thm:part_meta}.
\begin{lemma}
 \label{lem:potg_no_partitionable}
 Let $\G$ be the game defined in Lemma \ref{lem:potg_no_asymMeta}.
 Then, for any $\beta$ sufficiently large, $\G$ is not partitioned by the logit dynamics.
\end{lemma}
\begin{proof}
 The bottleneck of $(1, \ldots, 1)$ is asymptotically larger than any polynomial and smaller than any super-polynomial.
 Hence this profile cannot be contained in $N$ and it is not possible that $R_i = \{(1, \ldots, 1)\}$ for some $i$.
 Thus, it must be the case that $R_i = L$ for some subset $L \subseteq S$,
 such that $(1, \ldots, 1) \in L$, and $L \neq \{(1, \ldots, 1)\}$.
 Then, for $\beta$ sufficiently large $\pi_L(1, \ldots, 1) \leq 1/2$.
 But, since the bottleneck of $(1, \ldots, 1)$ is asymptotically larger than any polynomial,
 the mixing time of the chain restricted to $L$ is not polynomial.
\end{proof}

\subsection{Proof of Theorem \ref{thm:part_meta}}
A high level idea of the proof is discussed next.
We initially need to define the metastable distributions to which the logit dynamics converges.
To describe the distributions of interest for the given $n$-player 
game $\G$,
we leverage known results connecting bottleneck ratio and metastability.
In particular, it turns out that the stationary distribution of the dynamics restricted to of a subset of profiles
with bottleneck ratio at most the inverse of a super-polynomial, as defined in \eqref{eq:piL},
is metastable for a super-polynomial amount of time (Theorem~\ref{lem:meta:bottleneck}).
In this way we can easily ``build'' metastable distributions from the sets $R_i$
given by the definition of partitioned games (cf. Proposition~\ref{prop:meta_description}).

What about the pseudo-mixing time of this distribution?
We distinguish two cases. First we consider profiles that are in the ``\emph{core}'' of the support of this distribution,
namely the sets $T_i$ given by the definition of partitioned game.
We show that the pseudo-mixing time from these profiles is related to the mixing time of
the dynamics restricted to $R_i$ as described in \eqref{eq:restricted_loop} (see Corollary~\ref{cor:mix_ristrected_pseudo}).
Then, being the mixing time of the chains restricted to $R_i$ polynomial,
it follows that the pseudo-mixing time from the core is polynomial.

What about out-of-core profiles?
Suppose that there is a profile from which the dynamics takes long time to converge to a metastable distribution.
Then it must be the case that the dynamics takes long time to hit the core of one such distribution with high probability.
However, this cannot be the case since, by definition of partitioned games,
the logit dynamics from any non-core profile quickly hits a profile in the core of some distribution.

Next we formally prove Theorem~\ref{thm:part_meta}.
In particular, the proof follows from the Propositions~\ref{prop:meta_description}, \ref{prop:pmt_closeNash} and \ref{prop:pmt_farNash}
given below that, respectively, describe the metastable distributions, bound the pseudo-mixing time from the core,
and evaluate the behavior of the dynamics starting from non-core profiles.

\paragraph{Metastable distributions.}
We start by proving that some distributions defined on the sets $R_i$
are metastable for super-polynomial time.
\begin{prop}
\label{prop:meta_description}
Let $\G$ be a $n$-player game and consider the stationary distribution $\pi$ of the logit dynamics for $\G$.
If $\G$ is partitioned by the logit dynamics for $\beta$,
then for every $\varepsilon > 0$
there exists a function $\Time=\Time_\varepsilon$ at least super-polynomial in the input such that
for each%
\footnote{There can be values of $n$ for which the partition does not use the $i$-th ``component'' and thus $R_i, T_i$ and $\mu_i$ are not well defined. However, as long as there are infinite values of $n$ for which $R_i$ is given then asymptotic bounds on the metastability and the pseudo-mixing time of $\mu_i$ are well defined. Since the partition contains at least one ``component'' for any input, we have that there exists $n_0$ such that for $i \leq \max_{n \geq n_0} k(n)$, $R_i$ is defined infinite times.}
$i$,
and $n$ large enough,
the distribution $\mu_i$ that sets $\mu_i(\x) = \pi(\x) / \pi(R_i)$
is $(\varepsilon, \Time(n))$-metastable for the given $\beta$.
\end{prop}
\begin{proof}
Fix $i$. Given $\varepsilon > 0$, consider the function $\Time=\Time_\varepsilon$ such that $\Time(n) = \frac{\varepsilon}{B(R_i)} \geq \varepsilon q(n)$, where $R_i$ is the support of $\mu_i$. By the definition of $q$, $\Time$ is at least super-polynomial in the input.

By Lemma~\ref{lem:meta:bottleneck}, $\mu_i$ is $(B(R_i), 1)$-metastable. By Lemma~\ref{lem:meta:1}, $\mu_i$ is also $(B(R_i) \cdot \Time(n), \Time(n))$-metastable. The lemma follows since $B(R_i) \cdot \Time(n) = \varepsilon$.
\end{proof}

Finally, the following lemma shows that a combination of metastable distributions is metastable.
\begin{lemma}
 \label{lemma:convex_combo_metastable}
Let $P$ the transition matrix of a Markov chain with state space $S$ and  let $\mu_i$ be a distribution $(\varepsilon_i, \Time_i)$-met\-a\-sta\-ble for $P$, for $i = 1, 2, \ldots$. Set $\varepsilon = \max_i{\varepsilon_i}$ and $\Time = \min_i\{\Time_i\}$. Then, the distribution $\mu = \sum_{i} \alpha_i \mu_i$, with $\sum_{i} \alpha_i = 1$ and $\alpha_i \geq 0$, is $(\varepsilon, \Time)$-metastable.
\end{lemma}
\begin{proof}
 For every $t \leq \Time$ we have
 \begin{align*}
  \tv{\mu P^t - \mu} & = \max_{A \subseteq S} \left|(\mu P^t)(A) - \mu(A)\right| \\
  & = \max_{A \subseteq S} \left| \sum_{i} \alpha_i \left((\mu_i P^t)(A) - \mu_i(A)\right)\right| \\
  & \leq \sum_i \alpha_i \max_{A \subseteq S} \left|(\mu_i P^t)(A) - \mu_i(A)\right| \leq  \varepsilon.  \tag*{\qed}
\end{align*}
\let\qed\relax
\end{proof}

\paragraph{Pseudo-mixing time from the core.}
Now we prove that the logit dynamics for a $n$-player partitioned game converges in polynomial time
to the metastable distribution $\mu_i$ defined above, whenever the starting point is selected from the core $T_i$ of this distribution.
Specifically, we prove the following proposition.
\begin{prop}
 \label{prop:pmt_closeNash}
Let $\G$ be a $n$-player game and fix $\varepsilon > 0$.
If $\G$ is partitioned by the logit dynamics for $\beta$, then there is a function $p_\star$ at most polynomial in the input such that
for the given $\beta$, for each $i$ and for $n$ large enough, 
the pseudo-mixing time of $\mu_i$ from $T_i$ is $\pmt{\mu_i}{T_i}{2\varepsilon} = O(p_\star(n))$.
\end{prop}
In order to prove Proposition~\ref{prop:pmt_closeNash},
for $L \subseteq S$ non-empty, consider the Markov chain defined in \eqref{eq:restricted_loop}.
Let us abuse the notation and denote with $\mathring{P}_L$ and $\pi_L$ also the Markov chain and the distribution defined on the entire state space $S$, assuming $\mathring{P}_L(\x,\y) = 0$ if $\x \notin L$ or $\y \notin L$, and similarly $\pi_L(\x) = 0$ when $\x \notin L$.

For $L \subseteq S$ we set $\partial L$ as the border of $L$, that is the set of profiles in $L$ with at least a neighbor in $S \setminus L$.
Recall that $\tauset{S \setminus L}$ is the random variable denoting the first time the Markov chain with transition matrix $P$ hits a profile $\x \in S\setminus L$.
The following lemma formally proves the intuitive fact that, by starting from a profile in $L$ the chain $P$ and the chain $\mathring{P}_L$ are the same up to the time in which the former chain hits a profile in $S\setminus L$. The proof uses the well-known coupling technique (cf., e.g., \cite{lpwAMS08}) which is summarized in Appendix~\ref{sec:Markovcoupling}.
\begin{lemma}
\label{lemma:PeqP_R}
 Let $P$ be the transition matrix of a Markov chain with state space $S$ and let $\mathring{P}_L$ be the restriction of $P$ to $L \subseteq S$, $L \neq \emptyset$, as given in~\eqref{eq:restricted_loop}. Then, for every $\x \in L$ and for every $t > 0$,
 $$
  \tv{P^t(\x,\cdot) - \mathring{P}_L^t(\x, \cdot)} \leq \Prob{\x}{\tauset{S\setminus L} \leq t}.
 $$
\end{lemma}
\begin{proof}
Consider the following coupling $(X_t, Y_t)_{t>0}$ of the Markov chains with transition matrix $P$ and $\mathring{P}_L$, respectively:
 \begin{itemize}
  \item If $X_i = Y_i \in L \setminus \partial L$, then we update the first chain according to $P$ and obtain $X_{i+1}$; we then set  $Y_{i+1}=X_{i+1}$;
  \item If $X_i = Y_i \in \partial L$, then we update the first chain according to $P$: if $X_{i+1} \in L$, then we set $Y_{i+1} = X_{i+1}$, otherwise we set $Y_{i+1} = Y_i$;
  \item If $X_i \neq Y_i$, then we update the chains independently.
 \end{itemize}
 Since $X_0 = Y_0 = \x \in L$, we have that $X_t \neq Y_t$ only if $\tauset{S \setminus L} \leq t$. Thus, by the properties of couplings (see Theorem~\ref{thm:coupling}), we have
\[
  \tv{P^t(\x,\cdot) - \mathring{P}_L^t(\x, \cdot)} \leq \Prob{\x}{X_t \neq Y_t} \leq \Prob{\x}{\tauset{S \setminus L} \leq t}.  \tag*{\qed}
\]
\let\qed\relax
\end{proof}

The following corollary follows from the Lemma~\ref{lemma:PeqP_R} and the triangle inequality property of the total variation distance.
\begin{cor}
 \label{cor:mix_ristrected_pseudo}
 Let $P$ the transition matrix of a Markov chain with state space $S$ and let $\mathring{P}_L$ be the restriction of $P$ to a non-empty $L \subseteq S$ as given in \eqref{eq:restricted_loop}. Then, for every $\x \in L$ and for every $t > 0$,
 $$
  \tv{P^t(\x, \cdot) - \pi_L} \leq \tv{\mathring{P}^t_L(\x, \cdot) - \pi_L} + \Prob{\x}{\tauset{S\setminus L} \leq t}.
 $$
\end{cor}

Using Corollary~\ref{cor:mix_ristrected_pseudo} we can prove Proposition~\ref{prop:pmt_closeNash}.
\begin{proof}[Proof of Proposition~\ref{prop:pmt_closeNash}.]
For each $\x \in T_i$, by Corollary~\ref{cor:mix_ristrected_pseudo} and since $\Prob{\x}{\tauset{S\setminus R_i} \leq \tm^{R_i}(\varepsilon)} \leq \varepsilon$, we obtain
 $$
 \tv{P^{\tm^{R_i}(\varepsilon)}(\x, \cdot) - \mu_i} \leq \varepsilon + \varepsilon.
 $$
 The lemma follows from the observation that $\tm^{R_i}(\varepsilon)$ is at most a polynomial by definition of partitioned game.
\end{proof}

\paragraph{Pseudo-mixing time starting from the remaining profiles.}
Consider the distributions $\mu_i$ defined above (i.e., the stationary distribution restricted to $R_i$).
We focus here on the convergence time to distributions of the form
$$
\nu(\y) = \sum_{i} \alpha_i \mu_i(\y),
$$
for $\alpha_i \geq 0$ and $\sum_i \alpha_i = 1$.
Specifically,
for every profile $\x \in N$, we define the distribution
\begin{equation}
 \label{eq:nu_distro}
 \nu_\x(\y) = \sum_{i} \mu_i(\y) \cdot \Prob{\x}{X_{\tauset{S \setminus N}} \in T_i \mid \tauset{S \setminus N} \leq \Time_{S \setminus N}^\varepsilon(\x)},
\end{equation}
where $\Time_{S \setminus N}^\varepsilon(\x)$ is the first time step $t$ in which $\Prob{\x}{\tauset{S \setminus L} > t} \leq \varepsilon$.
Observe that by definition of $\tauset{S\setminus N}$, since the $T_i$'s and $N$ are a partition of $S$, $X_{\tauset{S\setminus N}} \in \cup_i T_i$ is a certain event for all values of $\tauset{S\setminus N}$.
Moreover, by the definition of $\Time_{S \setminus N}^\varepsilon(\x)$, the event $\tauset{S \setminus N} \leq \Time_{S \setminus N}^\varepsilon(\x)$ has non-zero probability and thus we can condition on it.
Thus, $\sum_i \Prob{\x}{X_{\tauset{S \setminus N}} \in T_i \mid \tauset{S \setminus N} \leq \Time_{S \setminus N}^\varepsilon(\x)}=1$.
The above is then a valid definition of the $\alpha_i$'s.

Then, we prove the following proposition.
\begin{prop}
\label{prop:pmt_farNash}
Let $\G$ be a $n$-player game and fix $\varepsilon > 0$.
If $\G$ is partitioned by the logit dynamics for $\beta$,
then there is a function $\Time=\Time_\varepsilon$ at least super-polynomial in the input and a function $p_\star$ at most polynomial in the input
such that for the given $\beta$, for every $n$ large enough and for each $\x \in N$
the corresponding distribution $\nu_\x$ is $(\varepsilon, \Time(n))$-metastable
and the pseudo-mixing time of $\nu_\x$ from the profile $\x$ is $\pmt{\nu_\x}{\{\x\}}{3\varepsilon} = O(p_\star(n))$.
\end{prop}
\begin{proof}
 Notice that, the distribution $\nu_\x$ is a convex combination of distributions that are metastable for super-polynomial time: thus, from Lemma~\ref{lemma:convex_combo_metastable}, there exist a function $\Time$ at least super-polynomial in the input such that each such $\nu_\x$ is $(\varepsilon, \Time(n))$-metastable.

Moreover, from the definition of partitioned game, we have that $\Time_{S \setminus N}^\varepsilon\left(\x\right)$
is at most a polynomial $\rho_\star(n)$ for any $\x \in N$.
Consider then the function $p_\star(\cdot)$ such that $p_\star(n) = \rho_\star(n) + \max_i \pmt{\mu_i}{T_i}{\varepsilon}$. From Proposition~\ref{prop:pmt_closeNash} and the fact that $\rho_\star$ is at most polynomial,
it turns out that $p_\star(\cdot)$ is at most a polynomial function in the input.
We complete the proof by showing that, for any sufficiently large $n$ and any $\x \in N$,
$p_\star(n)$ upper bounds the pseudo-mixing time $\pmt{\nu_\x}{\{\x\}}{3\varepsilon}$ to $\nu_\x$ from the profile $\x$.

We set $t^\star = p^\star(n)$, and denote with $E$ the event ``$\tauset{S \setminus N} \leq \Time_{S\setminus N}^\varepsilon(\x)$'' and with $\overline{E}$ its complement. Recall from Definition \ref{def:LD} that $X_t$ denotes the state of the Markov chain defined by logit dynamics at step $t$ and observe that
\begin{align*}
\tv{P^{t^\star}(\x,\cdot) - \nu_\x} & =  \max_{A \subset S} \left| \Prob{\x}{X_{t^\star} \in A} - \nu_\x(A)\right| \\
& = \max_{A \subset S} \left| \Prob{\x}{X_{t^\star} \in A \wedge E} - \nu_\x(A) + \Prob{\x}{X_{t^\star} \in A \wedge \overline{E}}\right| \\
& = \max_{A \subset S} \left| \Prob{\x}{X_{t^\star} \in A \mid E} (1 - \Prob{\x}{\overline{E}}) - \nu_\x(A) + \Prob{\x}{X_{t^\star} \in A \mid \overline{E}} \Prob{\x}{\overline{E}} \right| \\
& \leq \max_{A \subset S} \left| \Prob{\x}{X_{t^\star} \in A \mid E} - \nu_\x(A)\right| + \Prob{\x}{\overline{E}} \\
& \leq \tv{\Prob{\x}{X_{t^\star} \mid E} - \nu_\x}+\varepsilon,
\end{align*}
where the definition of $\Time_{S\setminus N}^\varepsilon(\x)$ implies that $\Prob{\x}{E} \geq 1- \varepsilon >0$ and then yields the third equality and last inequality. The penultimate inequality, instead, simply follows from the subadditivity of the absolute value and the fact that the difference between two probabilities is upper bounded by $1$. As every $\mu_i$ is metastable for at least a super-polynomial number of steps, we have, by using $\tau^\star$ as a shorthand for $\tauset{S \setminus N}$,
\begin{align*}
\tv{\Prob{\x}{X_{t^\star} \mid E} - \nu_\x} & =  \tv{\sum_i \sum_{\y \in T_i} \Prob{\x}{X_{\tau^\star} = \y \mid E} \cdot \Prob{\x}{X_{t^\star} \mid X_{\tau^\star} = \y \wedge E} - \nu_\x}\\
& \leq \tv{\sum_{i} \sum_{\y \in T_i} \Prob{\x}{X_{\tau^\star} = \y \mid E} \left(P^{t^\star - \tau^\star}(\y, \cdot) - \mu_i \right)}\\
& \leq  \sum_{i} \sum_{\y \in T_i} \Prob{\x}{X_{\tau^\star} = \y \mid E} \tv{P^{t^\star - \tau^\star}(\y, \cdot) - \mu_i} \leq  2\varepsilon,
\end{align*}
where the definition of $\tau^\star$ yields $X_{\tau^\star} \in T_i$, for some $i$, which in turns yields the first equality by the law of total probability. In the first inequality above, instead, we use the definition of $\nu_{\x}$ and the fact that by definition of $t^\star$, $E$ implies $t^\star - \tau^\star \geq t^\star - \Time_{S\setminus N}^\varepsilon(\x) \geq \max_i \pmt{\mu_i}{T_i}{\varepsilon}$; the second inequality follows from a simple union bound; and the last inequality follows from Lemma~\ref{lem:metaonetot} (note that $t^\star-\tau^\star$ satisfies the hypothesis of the lemma: the lower bound is showed above, while the upper bound follows from the fact that the $\mu_i$'s are metastable for at least super-polynomial time).
Hence, we have for every sufficiently large $n$ and every $\x \in N$, $\pmt{\nu_\x}{\{\x\}}{3\varepsilon} \leq t^\star = p_\star(n)$.
\end{proof}

\section{Asymptotically well-behaved potential games}
\label{sec:awd}
We now ask what class of potential games are partitioned by the logit dynamics. We know already that the answer must differ from the whole class of potential games, due to Lemma \ref{lem:potg_no_asymMeta}.
However, it is important to understand to what extent it is possible to prove asymptotic metastability for potential games.

Our main aim is to give results, asymptotic in the number $n$ of players, about the behavior of logit dynamics run on potential games.
Clearly, it makes sense to give asymptotic results about the property of an object, only if this property is asymptotically well-defined, that is, the object is uniquely defined for infinitely many values of the parameter according to which we compute the asymptotic and the property of this object does not depend chaotically on this parameter.

For example when we say that a graph has large expansion,
we actually mean that there is a sequence of graphs indexed by the number of vertices,
such that the expansion of each graph can be bounded by a single function of this number.
Similarly, when we say that a Markov Chain has large mixing time,
we actually mean that there is a sequence of Markov chains indexed by the number of states, such that the mixing time of each Markov chain can be bounded by a single function of this number. Yet another example arises in algorithm game theory: when we say that the Price of Anarchy of a game is large, we actually mean that there is a sequence of games indexed, for example, by the number of players such that the Price of Anarchy of each game can be bounded by a single function of this number\footnote{ Sometimes Price of Anarchy results do not refer to a single game, but to a class of games (e.g., congestion games). In this case, we can explain the asymptotic result according to two different viewpoints. Indeed, we can see the game ``horizontally'', that is as a sequence of sets of instances, where each set contains all instances defined for some specified number of players. Then the objects to which the asymptotic result refers is given by these sets of instances and we require that worst Price of Anarchy in each set is bounded by a single function of the number of players. Similarly, we can see the game ``vertically'', that is as a set of games, where each game contains at most one instance for each number of players. Then the objects are these games, and the bound on the Price of Anarchy is just the worst case bound among all games in the set. In this work, we will assume this second viewpoint. 
}.

In this work the object of interest is a potential game
and the property of interest is the behavior of the logit dynamics for this game.
And thus, in our setting, it makes sense to give asymptotic results
only when a potential game is uniquely defined for infinitely many $n$
and the behavior of the logit dynamics for the potential game is not chaotic as $n$ increases.
However, giving a formal definition of what this means is not as immediate
as in the case of the expansion of a graph or of the Price of Anarchy of a game.
Thus, in order to gain insight on how to formalize this concept, let us look at
some examples of games for which the behavior of the logit dynamics is evidently ``almost the same'' as $n$ increases
(these include the examples of partitioned games analyzed above) and
examples in which this behavior instead changes infinitely often.

\paragraph{\cw-game.} The behavior of logit dynamics for this game can be described in a way that is immaterial from the actual value of $n$.
Indeed, the potential function has two equivalent minima when either all players adopt strategy $-1$ or all players adopt strategy $+1$
and it increases as the number of players adopting a strategy different from the one played by the majority of agents increases. The potential reaches its maximum when each strategy is adopted by the same number of players.
Moreover, regardless of the actual value of $n$ it is easy to see that if the number of $+1$ strategies
is sufficiently larger than the number of $-1$ strategies\footnote{The
extent to which the number of $+1$ must be larger than the number of $-1$
can depend on $n$, but it can be bounded by a single function $F$ on the number of players.},
then it must be hard for the logit dynamics to reach a profile with more $-1$'s than $+1$'s,
whereas it must be easy for the dynamics to converge to the potential minimizer in which all players are playing $+1$.

Thus the logit dynamics for the Curie-Weiss game is asymptotically well-behaved for our purposes, since the evolution of the dynamics can be described in a way that is immaterial of the actual value of $n$.
In other words, the evolution of the dynamics when the number of players is $n$, can be mapped into the evolution of the dynamics for smaller or larger numbers of players,
such that the time necessary to some events to happen (e.g.,
the time for reaching or leaving certain subsets of the profile space)
can be bounded by the same function of the number of players.

A similar argument holds even for the other partitioned games described in Section~\ref{subsec:ex_part}.

\paragraph{Pigou's congestion game.} Another example of game for which it is immediate to see that the behavior of the logit dynamics does not change as
the number of players increases is the well-known Pigou's congestion game:
there are two links, one has fixed cost $1$, whereas the second one has congestion-dependent cost $c/n$,
where $c$ denotes the number of players choosing this second link.
It is well known that this game is a potential game \cite{R73},
with potential function $\Phi(\x) = \left[(n-c) + \frac{1}{n} \sum_{i = 1}^c i\right] = \frac{1}{2n} \left(2n^2 - 2nc + c(c + 1)\right)$.

As above, the behavior of this game can be easily described in a way that is immaterial from the actual value of $n$.
Indeed, it is easy to check that the potential function is minimized in the profile in which all players are using the second link,
i.e., the one with congestion-dependent cost, it is maximized in the profile in which all players are using the fixed-cost link,
and, in general, the potential decreases as the number of players adopting the second link increases.
Thus, it happens that the logit dynamics for this game quickly converges from any starting profile to the profile in which 
all players are on the second link regardless of the value that $n$ actually assumes\footnote{How
quick the convergence is depends on the value of $n$, but this convergence time can be bounded by the same small function of $n$.}.

It is then evident that the behavior of the logit dynamics for Pigou's example is asymptotically well-defined
and  the evolution of the dynamics when the number of players is $n$
can be mapped in the evolution of the dynamics for smaller or larger numbers of players,
so that the time for reaching or leaving certain subsets of the profile space
can be bounded by the same function of the number of players.

\paragraph{A game for which the logit dynamics chaotically depends on $n$.} Observe that, even though the potential function of the game in Lemma \ref{lem:potg_no_asymMeta} can be easily described solely as a function of $n$,
just as done for previous games, we cannot describe how the logit dynamics for this game
and a given rationality parameter $\beta$ behaves as $n$ changes.
In particular, we are unable to describe through a single function
the time that is necessary for the logit dynamics to leave the profile in which all players are adopting strategy $1$.
This time, by construction, changes infinitely often and,
for any tentative bound, there will always be a value of $n$
from which that bound will turn out to be not correct.

\paragraph{Asymptotically well-behaved games: the definition.}
From the analysis of these games, it is evident that the behavior of the logit dynamics for potential games
is asymptotically well-defined when profiles of the $n$-player game can be associated to profiles of the $n'$-player game
such that the probability of leaving associated profiles can be always bounded by the same function of the number of players.
Formally, we have the following definition.
\begin{definition}
The logit dynamics for a (potential) game is asymptotically well-behaved  if there is $n_0$ such that for every $n \geq n_0$ and for every $L' \subseteq S_1 \times \cdots \times S_{n}$
  there is a subset $L \subseteq S_1 \times \cdots \times S_{n_0}$ and a function $F_L$ such that
  $B(L) = \Theta(F_L(n_0))$ and $B(L') = \Theta(F_L(n))$.
\end{definition}
For sake of compactness, we will simply say that the potential game is asymptotically well-behaved whenever
the logit dynamics (run on it) is. We also call the constant $n_0$ in the above definition the \emph{asymptotic basis} of the game.


The main result of our paper follows.
\begin{theorem}\label{thm:awd_part}
Let $\G$ be an asymptotically well-behaved $n$-player potential game and let $\Delta(n)$ be the function that, for every $n$, gives the Lipschitz constant of the potential function $\Phi$ for a game $\G$ with $n$ players, i.e.,
$$
 \Delta(n) := \max \left\{ \Phi(\x) - \Phi(\y) \colon H(\x,\y) = 1 \right\}.
$$
Then, for any function $\rho$ at most polynomial in its input,
$\G$ is asymptotically metastable for each $\beta_0 \leq \beta \leq \frac{\rho(n)}{\Delta(n)}$.
\end{theorem}
The proof builds upon Theorem~\ref{thm:part_meta} and proves that 
any asymptotic well-behaved potential game is partitioned by the logit dynamics for $\beta$ not too large as in the statement.

\subsection{Proof of Theorem \ref{thm:awd_part}}
\label{sec:sufficient}
We begin by highlighting a property of asymptotically well-behaved games will turn out to be very important below. 
\begin{lemma}
 \label{lem:awd_classified}
 Let $\G$ be an asymptotically well-behaved $n$-player potential game and let $n_0$ be its asymptotic basis.
 Then, there is a polynomial $p$ and a super-polynomial $q$ such that
 for any $n \geq n_0$ and for every $L' \subseteq S_1 \times \cdots \times S_{n}$,
 either $B(L') \geq \frac{1}{p(n)}$  or $B(L') \leq \frac{1}{q(n)}$.
\end{lemma}
\begin{proof}
 By definition of asymptotically well-behaved game, there is a subset $L \subseteq S = S_1 \times \cdots \times S_{n_0}$
 and a function $F_L$, such that $B(L') = \Theta(F_L(n))$.
 Since there are a finite number (specifically, at most $m^{n_0}$) of such functions $F_L$,
 we can always separate polynomials and super-polynomials ones.
 That is, we can partition the subsets of $S$ in two (possibly empty) subsets $S'$ and $S''$
 such that for any $L$ in $S'$, $F_L(n) \geq 1/p(n)$ and, for any $L$ in $S''$, $F_L(n) \leq 1/q(n)$,
 for any $n \geq n_0$ and for $p$ polynomial and $q$ super-polynomial.
\end{proof}
Henceforth, we will say that the functions $p, q$ are \emph{generated} by the game $\G$.

We now introduce a (computationally infeasible) algorithm that computes
subsets $R_1, \ldots, R_k$ of $S$ and a partition $T_1, \ldots, T_k, N$ of $S$. Next we show that under the condition that the potential game is asymptotically well-behaved,
the sets returned by this algorithm enjoy the properties required by the definition of partitioned games.

The procedure works its way by finding subsets of profiles that act as super-polynomial bottlenecks for the Markov chain.
The algorithm $\A_{p,q}$ takes in input an $n$-player asymptotically well-behaved potential game $\G$, a rationality level $\beta$, a constant $\varepsilon > 0$ and $n$; it is parametrized by two functions $p$ at most polynomial and $q$ at least super-polynomial generated by $\G$.
\begin{algorithm}[$\A_{p,q}$]
\label{algo:meta_sets}
Set $N = S$ and $i = 1$. While there is a set $L \subseteq N$ with $\pi(L) \leq 1/2$ such that $B(L) \leq 1/q(n)$, do:
\begin{enumerate}
 \item Denote with $R_i$ one such subset with the smallest stationary probability;
 \item Denote with $T_i$ the largest subset of $R_i$ such that for every $\y \in T_i$,
 $$\Prob{\y}{\tauset{S\setminus R_i} \leq \tm^{R_i}(\varepsilon)} \leq \varepsilon;$$
 \item If $T_i$ is not empty, return $R_i$ and $T_i$, delete from $N$ all profiles contained in $T_i$ and increase $i$. Otherwise, terminate the algorithm.
\end{enumerate}
\end{algorithm}
Observe that if there is a disconnected set $L$ such that $B(L) \leq 1/ q(n)$, then each connected component $L'$ of $L$ will have $B(L') \leq 1/ q(n)$ and smaller stationary probability: hence, the set $R_i$ returned by the algorithm will be connected.
Note also that by Theorem~\ref{theorem:bottleneck} and the assumption that we are considering only cases in which the mixing time is super-polynomial,
the algorithm above enters at least once in the loop (and thus at least a subset $R_i$ is computed).

Clearly the sets $R_i$ returned by the algorithm enjoy the property of having super-polynomial bottleneck ratio
and the sets $T_i$ satisfy the requirement that, starting from any $\x \in T_i$, it is unlikely to leave $R_i$ quickly.
It is then left to prove that the mixing time of the chains restricted to $R_i$ is polynomial
and that it is easy to leave the set $N$.
This follows from the following propositions that are proved in the next sections.
\begin{prop}
 \label{prop:mixing_restricted_chain}
Let $\G$ be an asymptotically well-behaved $n$-player potential game; fix $\varepsilon > 0$ and a function $\rho$ at most polynomial in its input.
Let $p, q$ be the functions generated by $\G$.
Consider the sequence of sets $R_i$ returned by $\A_{p,q}$.
For any $n$ sufficiently large, if $\beta_0 \leq \beta \leq \frac{\rho(n)}{\Delta(n)}$ then $\tm^{R_i}(\varepsilon)$ is at most polynomial.
\end{prop}
\begin{prop}
 \label{prop:hitting_outN}
Let $\G$ be an asymptotically well-behaved $n$-player potential game; fix $\varepsilon > 0$ and a function $\rho$ at most polynomial in its input.
Let $p, q$ the functions generated by $\G$ and recall that, for $\x \in L$ and $0 < \varepsilon < 1$, $\Time_{S \setminus L}^\varepsilon(\x)$
is defined as the first time step $t$ in which $\Prob{\x}{\tauset{S \setminus L} > t} \leq \varepsilon$.
Consider the set $N$ returned by $\A_{p,q}$.
For any $n$ sufficiently large, if $\beta_0 \leq \beta \leq \frac{\rho(n)}{\Delta(n)}$ then $\Time_{S \setminus N}^\varepsilon(\x)$ is at most polynomial.
\end{prop}
This proves that asymptotically well-behaved potential games are partitioned by the logit dynamics
and, hence, they are asymptotically metastable, concluding the proof of Theorem~\ref{thm:awd_part}.

\subsection{Proof of Proposition~\ref{prop:mixing_restricted_chain}}
A high level idea of the proof of Proposition~\ref{prop:mixing_restricted_chain} is discussed next.
We first give a spectral characterization of the transition matrix defined in \eqref{eq:restricted_loop} (Lemma~\ref{lemma:diod_conj_restricted}).
Then we show that no subset of $R_i$ has small bottleneck ratio in the restricted chain (Lemma~\ref{lemma:bottleneck_poly}).
Note that this does not directly follow from $R_i$ being the smaller possible subset with super-polynomial bottleneck ratio.
Indeed, the bottleneck ratio of a subset depends on the dynamics according to which it is computed.
Thus, a subset can have a small bottleneck ratio when computed within the reference frame of the restricted dynamics,
but not when we refer to the original dynamics.
Nevertheless, we show that this is not the case. Specifically,
we will show that for asymptotically well-behaved games
there is a relationship between the bottleneck ratio of a subset of profiles in the restricted and in the original dynamics.
Finally, the result follows from the known relationship among mixing time,
relaxation time and bottleneck ratio (see Theorems~\ref{theorem:relaxation} and \ref{thm:bottle_rel}).

\subsubsection{Spectral property of logit dynamics restrictions}
In~\cite{afpppSPAA11j} it has been shown that all the eigenvalues of the transition matrix of logit dynamics for potential games are non-negative.
The technique used in that proof can be generalized to work also for some restrictions of these matrices.

To begin, we note that the definition of reversibility can be extended in a natural way to any square matrix and probability distribution over the set of rows of the matrix. We then state a fairly standard result relating eigenvalues of matrices to certain inner products.
\begin{lemma}\label{le:innerproducts}
Let $P$ be a square matrix on state space $S$ and $\pi$ be a probability distribution on $S$. If $P$ is reversible with respect to $\pi$ and has no negative eigenvalues then for any function $f:S \rightarrow \mathbb{R}$ we have
$$
\langle Pf,f \rangle_\pi := \sum_{\x \in S} \pi(\x) (Pf)(\x) f(\x) \geq 0.
$$
\end{lemma}
\begin{proof}
Let $\lambda_1, \ldots, \lambda_s$, $s = |S|$, be the eigenvalues of $P$.
Moreover, let $f_1, \ldots, f_s$ denote their corresponding eigenfunctions. For any $\x \in S$, we then have $(Pf_i)(\x) f_i(\x)=\lambda_i f_i(\x)$. Since $P$ is reversible then we know that the eigenfunctions assume real values and that they form an orthonormal basis for the space $(\mathbb{R}^s, \langle \cdot, \cdot \rangle_\pi)$ (see, e.g., Lemma 12.2 in \cite{lpwAMS08}). Then any real-valued function $f$ defined upon $S$ can be expressed as a linear combination of the $f_i$'s. Thus, there exist  $\alpha_i$'s in $\mathbb{R}$ such that
\[
\sum_{\x \in S} \pi(\x) (Pf)(\x) f(\x) = \sum_{\x \in S} \pi(\x) \sum_{i=1}^s \alpha_i^2 (P f_i)(\x) f_i(\x) = \sum_{\x \in S} \pi(\x) \sum_{i=1}^s \alpha_i^2 \lambda_i f_i^2(\x) \geq 0. \tag*{\qed}
\]
\let\qed\relax
\end{proof}

To specify the restrictions of the transition matrix we are interested in, let $\G$ be a game with profile space $S$ and let $P$ be the transition matrix of the logit dynamics for $\G$; we say that a $|A| \times |A|$ matrix $P'$, with $A \subseteq S$, is a \emph{nice restriction} of $P$ if there exists $L \subseteq A$, $L \neq \emptyset$, such that $P'(\x,\x)\geq P(\x,\x)$ for $\x \in L$, $P'(\x,\y) = P(\x,\y)$ if $\x,\y \in L$, $\x \neq \y$, and is $0$ otherwise. Note that $P$ is a nice restriction of itself. We generalize the result given in~\cite{afpppSPAA11j} to nice restrictions of the transition matrix of logit dynamics for potential games.
\begin{lemma}
 \label{lemma:diod_conj_restricted}
Let $\mathcal{G}$ be a game with profile space $S$, let $P$ be the transition matrix of the logit dynamics for $\mathcal{G}$ and let $P'$ be a nice restriction of $P$ with state space $A$. If $P$ is reversible then no eigenvalue of $P'$ is negative.
\end{lemma}
\begin{proof}
Firstly, note that if $P$ is reversible with respect to $\pi$ then the nice restriction $P'$, defined upon a subset of states $A$, is reversible with respect to $\pi'$ defined as $\pi$ restricted to $A$, i.e., $\pi'(\x)=\pi(\x)/\pi(A)$ for $\x \in A$.

Assume for sake of contradiction that there exists an eigenvalue $\lambda<0$ of $P'$. Let $f_\lambda$ be an eigenfunction of $\lambda$. Note that since $P$ is reversible then $f_\lambda$ is real-valued. By definition, $f_\lambda \neq \0$; hence, since $\lambda < 0$ and as $(P'f_\lambda)(\x)= \lambda f_\lambda(\x)$, then for every profile $\x \in A$ such that $f_\lambda(\x) \neq 0$ we have $\sign{(P'f_\lambda)(\x)} \neq \sign{f_\lambda(\x)}$ and thus
 $$
  \langle P'f_\lambda, f_\lambda \rangle_{\pi'} = \sum_{\x \in A} \pi'(\x) (P'f_\lambda)(\x) f_\lambda(\x) < 0.
 $$

Let $L$ denote the maximal subset of $A$ for which $P'$ is a nice restriction of $P$. Let us denote with $P^L$ the transition matrix on the state space $A$ such that $P^L(\x,\y) = P(\x,\y)$ for every $\x,\y \in L$ and $P^L(\x,\y)=0$ otherwise. Then we can write $P'$ as $P^L + (P' - P^L)$: by the definition of nice restriction $(P' - P^L)$ is a non-negative diagonal matrix. Therefore, $(P' - P^L)$ is reversible with respect to $\pi'$. Since the eigenvalues of a diagonal matrix are exactly the diagonal elements, we have that $(P' - P^L)$ has non-negative eigenvalues and then, by Lemma \ref{le:innerproducts},
$\langle (P'-P^L)f_\lambda, f_\lambda \rangle_{\pi'} \geq 0$.
Moreover, for every $i$ and for every $\z_{-i}$, we denote with $P_{i,\z_{-i}}$ the matrix such that for every $\x,\y \in A$
 \begin{equation}
 \label{eq:Pizi}
  P_{i, \z_{-i}}(\x,\y) = \frac{1}{nZ_i(\z_{-i})}\begin{cases}
          e^{\beta u_i(\y)}, & \mbox{if } \x_{-i} = \y_{-i} = \z_{-i} \mbox{ and } \x,\y \in L;\\
	  0, & \mbox{otherwise}.
         \end{cases}
 \end{equation}
Observe that $P_{i, \z_{-i}}$ has at least one non-zero row and that all non-zero rows of $P_{i, \z_{-i}}$ are the same. Thus $P_{i, \z_{-i}}$ has rank 1, and hence since it is a non-negative matrix all its eigenvalues are non-negative \cite{hj90}\footnote{This result about the eigenvalues of matrices with rank 1 appears as an exercise at page 61 of \cite{hj90} and in \cite{Osn05}.}. Moreover, since all off-diagonal entries of $P_{i, \z_{-i}}$ are either $0$ or equal to the corresponding entry of $P'$ we can conclude that $P_{i, \z_{-i}}$ is reversible with respect to $\pi'$. Thus, Lemma \ref{le:innerproducts} yields $\langle P_{i, \z_{-i}}f_\lambda, f_\lambda \rangle_{\pi'} \geq 0$.
Finally, observe that $P^L = \sum_i \sum_{\z_{-i}} P_{i, \z_{-i}}$. Hence from the linearity of the inner product, it follows that $\langle P'f_\lambda, f_\lambda \rangle_{\pi'} \geq 0$ and thus we reach a contradiction.
\end{proof}

It is immediate to see that the restricted chain $\mathring{P}_L$ defined in \eqref{eq:restricted_loop}
is a nice restriction of $P$ and hence all its eigenvalues are non-negative by the theorem above.

\subsubsection{Mixing time of the restricted chains}
\label{subsec:mixing}
Before bounding the mixing time of the restricted chain we prove a very important preliminary lemma.
\begin{lemma}
\label{lemma:bottleneck_poly}
Let $\G$ be an asymptotically well-behaved $n$-player potential game and fix $\beta \geq \beta_0, \varepsilon > 0$.
Let $p, q$ the functions generated by $\G$.
Consider the sequence of sets $R_i$ returned by $\A_{p,q}$. Then, for $n$ sufficiently large and for any $A \subseteq R_i$, we have
$$B_{R_i}(A) \geq \frac{1}{p(n)} - \frac{1}{\ell(n)},$$
where $\ell$ is at least super-polynomial.
\end{lemma}
\begin{proof}
Let us postpone the exact definition of $\ell$ and suppose, by contradiction, that there are infinitely many $n$ for which there is $A \subset R_i$ such that $B_{R_i}(A) < \frac{1}{p(n)} - \frac{1}{\ell(n)}$.

We will show that for $n$ sufficiently large either $B\left(A\right)\leq 1/q(n)$ or $B(\overline{A})\leq 1/q(n)$, where $\overline{A} = R_i \setminus A$. Then, since they are contained in $R_i$ and hence their stationary probability is less than $\pi\left(R_i\right)$, one of these set must be chosen before $R_i$ by $\A_{p,q}$. But since in the third step of Algorithm~\ref{algo:meta_sets} either at least one element of such sets should be deleted from $N$ or the algorithm terminates, as a consequence, we have that $R_i$ cannot be returned by the algorithm, thus a contradiction.

Consider the function $v(\cdot)$ that sets $v(n)= \frac{\pi\left(A\right)}{Q\left(A, S \setminus R_i\right)}$. We distinguish two cases depending on how $v$ evolves as $n$ grows.

\medskip \noindent \underline{If $v(\cdot)$ is at least super-polynomial in the input:}
We have
 \begin{align*}
  B\left(A\right) & = \frac{Q\left(A, S \setminus A\right)}{\pi\left(A\right)} = \frac{Q\left(A, R_i \setminus A\right)}{\pi\left(A\right)} + \frac{Q\left(A, S \setminus R_i\right)}{\pi\left(A\right)}\\
  & = \frac{\sum_{\x \in A} \sum_{\y \in R_i \setminus A} \pi(\x) P(\x,\y)}{\pi\left(A\right)}+ \frac{Q\left(A, S \setminus R_i\right)}{\pi\left(A\right)} \\
  & =  \frac{\sum_{\x \in A} \sum_{\y \in R_i \setminus A} \pi_{R_i}(\x) \mathring{P}_{R_i}(\x,\y)}{\pi_{R_i}\left(A\right)} + \frac{Q\left(A, S \setminus R_i\right)}{\pi\left(A\right)}\\
  & = B_{R_i}\left(A\right) + \frac{Q\left(A, S \setminus R_i\right)}{\pi\left(A\right)}< \frac{1}{p(n)} + \frac{1}{v(n)} - \frac{1}{\ell(n)}.
 \end{align*}
 By taking $\ell(n) \leq v(n)$ for each $n$ sufficiently large, we have that $B\left(A\right) < \frac{1}{p(n)}$.
 Then, since $\G$ is asymptotically well-behaved, from Lemma~\ref{lem:awd_classified} it follows that $B\left(A\right) \leq \frac{1}{q(n)}$.
 
\medskip \noindent \underline{If $v(\cdot)$ is polynomial in the input:}
Note that $\frac{Q\left(A, S \setminus R_i\right)}{\pi\left(R_i\right)} + \frac{Q\left(\overline{A}, S \setminus R_i\right)}{\pi\left(R_i\right)} = B\left(R_i\right) \leq \frac{1}{q(n)}$, otherwise $R_i$ was not returned by the algorithm. Hence, we obtain
$$
 Q\left(A, S \setminus R_i\right) \leq \frac{1}{q(n)} \cdot \pi\left(R_i\right) \qquad \text{and} \qquad Q\left(\overline{A}, S \setminus R_i\right) \leq  \frac{1}{q(n)} \cdot \pi\left(R_i\right).
$$
From the first of these inequalities, we have
$\pi\left(A\right) \leq \frac{v(n)}{q(n)} \cdot \pi\left(R_i\right)$.
Hence
$$\frac{Q\left(A, \overline{A}\right)}{\pi\left(R_i\right)} \leq \frac{v(n)}{q(n)} \cdot \frac{Q\left(A, \overline{A}\right)}{\pi\left(A\right)} = \frac{v(n)}{q(n)} \cdot B_{R_i}\left(A\right) < \frac{v(n)}{q(n)} \left(\frac{1}{p(n)} - \frac{1}{\ell(n)}\right).
$$
Then we obtain
\begin{align*}
 B\left(\overline{A}\right) & = \frac{Q\left(\overline{A}, S \setminus \overline{A}\right)}{\pi\left(\overline{A}\right)} = \frac{Q\left(\overline{A}, A\right)}{\pi\left(R_i\right) - \pi\left(A\right)} + \frac{Q\left(\overline{A}, S \setminus R_i\right)}{\pi\left(R_i\right) - \pi\left(A\right)}\\
 \text{(by reversibility of $P$)} \quad & = \frac{Q\left(A, \overline{A}\right)}{\pi\left(R_i\right) - \pi\left(A\right)} + \frac{Q\left(\overline{A}, S \setminus R_i\right)}{\pi\left(R_i\right) - \pi\left(A\right)}\\
 & \leq \frac{v(n)}{q(n)} \left(\frac{1}{p(n)} - \frac{1}{\ell(n)}\right) \left(1 - \frac{v(n)}{q(n)}\right)^{-1} + \frac{1}{q(n)} \left(1 - \frac{v(v)}{q(n)}\right)^{-1} \\ & = O\left(\frac{1}{q(n) - v(n)}\right),
\end{align*}
where the upper bounds hold for each choice of super-polynomial function $\ell$. Since $q(n) - v(n)$ evolves at least as a super-polynomial, if $n$ is sufficiently large, $B(\overline{A}) < \frac{1}{p(n)}$.
Then, since $\G$ is asymptotically well-behaved, from Lemma~\ref{lem:awd_classified} it follows that $B\left(\overline{A}\right) \leq \frac{1}{q(n)}$.
\end{proof}

Now we are ready to prove the mixing time of the chain restricted to $R_i$ is polynomial.
\begin{proof}[Proof of Proposition~\ref{prop:mixing_restricted_chain}]
Consider the set of profiles $A_\star \subset R_i$ that minimizes $B_{R_i}\left(A\right)$ among all $A \subset R_i$ such that $\pi_{R_i}\left(A\right) \leq 1/2$. By Lemma~\ref{lemma:bottleneck_poly}, $B_{R_i}\left(A_\star\right) \geq 1/p(n) - 1/\ell(n)$ for each $n$ sufficiently large.

Moreover, for each $n$ and each $\x \in R_i$, since $|S| \leq m(n)^n$, it follows that
$$
 \log \frac{1}{\pi_{R_i}\left(\x\right)} \leq \log \frac{|S| e^{-\beta \Phi_{\min}}}{e^{-\beta \Phi_{\max}}} \leq \log \frac{e^{n \log m(n)} e^{-\beta \Phi_{\min}}}{e^{-\beta \Phi_{\max}}} = n \log m(n) + \beta \left(\Phi_{\max} - \Phi_{\min}\right),
$$
where $\Phi_{\max}$ and $\Phi_{\min}$ denote the maximum and minimum of the potential $\Phi$ overall possible strategy profiles.
Since $\Phi_{\max} - \Phi_{\min} \leq n \cdot \Delta(n)$ and $\beta \leq \rho(n) / \Delta(n)$, then
$$
 \log \frac{1}{\pi_{R_i}\left(\x\right)} \leq n \cdot \left(\log m(n) + \rho(n)\right).
$$

Then, since $\left(\frac{1}{p} - \frac{1}{\ell}\right) = \Theta\left(\frac{1}{p}\right)$, from Lemma~\ref{lemma:diod_conj_restricted} and the properties of the relaxation time (see Theorems~\ref{thm:bottle_rel} and \ref{theorem:relaxation}) it follows that the mixing time is
 $$
  \tm^{R_i}(\varepsilon) \leq \left(\frac{1}{p(n)} - \frac{1}{\ell(n)}\right)^{-2} \cdot \left(n \log m(n) + \rho(n)\right) \cdot 2 \log \frac{4}{\varepsilon} = O(p_\star(n)).
 $$
Since $p$, $\log m$ and $\rho$ are at most polynomial, then $p_\star$ is at most polynomial in its input and the lemma follows.
\end{proof}

\subsection{Proof of Proposition~\ref{prop:hitting_outN}}
In order to prove Proposition~\ref{prop:hitting_outN},
we show that there is a strong relationship between hitting time and metastability (see Lemma~\ref{lemma:hitting_bottle_ub} and Lemma~\ref{lemma:hitting_bottle_lb}) and, in particular, that high hitting time implies the existence of a subset with small bottleneck ratio.
Note that $N$ contains subsets of small bottleneck ratio only if at some iteration $T_i$ is empty.
Therefore, it is sufficient to prove that the cores are not empty for asymptotically well-behaved games (see Lemma~\ref{lemma:hitting}).

\subsubsection{The relation between bottleneck ratio and hitting time}
\label{sec:meta_preliminaries}
For a game $\G$ with potential function $\Phi$ and profile space $S$, and a rationality level $\beta$, let $P$ be the transition matrix of the Markov chain defined by the logit dynamics on $\G$. For a non-empty $L \subseteq S$, we denote with $P_{\overline{L}}$ the matrix
\begin{equation}
 \label{eq:restricted}
 P_{\overline{L}}(\x,\y) = \begin{cases}
                          P(\x,\y) & \text{if } \x,\y \in L;\\
                          0 & \text{otherwise}.
                         \end{cases}
\end{equation}
Let $\lambda^{\overline{L}}_1 \geq \lambda^{\overline{L}}_2 \geq \ldots \geq \lambda^{\overline{L}}_{|S|}$ be the eigenvalues of $P_{\overline{L}}$: notice that $\lambda^{\overline{L}}_1$ can be different from $1$ since the matrix $P_{\overline{L}}$ is not stochastic. Lemma~\ref{lemma:diod_conj_restricted} implies that  $\lambda^{\overline{L}}_1 \geq \lambda^{\overline{L}}_2 \geq \ldots \geq \lambda^{\overline{L}}_{|S|} \geq 0$, and thus for $\lambda^{\overline{L}}_{\max}$, the largest eigenvalue of $P_{\overline{L}}$ in absolute value, we have: $\lambda^{\overline{L}}_{\max} = \max_i |\lambda^{\overline{L}}_i| = \lambda^{\overline{L}}_1$.

We start with two characterizations of $1 - \lambda^{\overline{L}}_{\max}$ in terms of bottleneck ratio.
The first one is an easy extension of the similar characterization of the spectral gap of stochastic matrices.
\begin{lemma}
 \label{lemma:lambda_max_bottle}
 For finite $\beta$ and any $\emptyset \neq L \subseteq S$,
 $
  1 - \lambda^{\overline{L}}_{\max} \leq B(L).
 $
\end{lemma}
\begin{proof}
 Define the function $\varphi_L: S \rightarrow [0,1]$ to be such that $\varphi_L(\x) = \pi(L)$ if $\x \in L$, and $\varphi_L(\x) = 0$ otherwise. Consider now the function
 \begin{equation}\label{eq:def:dir}
 \Dir_P(\varphi_L) := \frac{1}{2} \sum_{\x, \y \in S} \pi(\x) P(\x,\y) (\varphi_L(\x) - \varphi_L(\y))^2.
 \end{equation}
 By Theorem~\ref{thm:gibbsPot}, $\pi(L) \neq 0$ and then $\Expec{\pi}{\varphi_L^2} = \pi(L)^3 \neq 0$. Moreover, by recalling the definition of $\partial L$ as the set of profiles $\x \in L$ that have at least one neighbor profile in $S \setminus L$ and denoting with $E(A_1,A_2)$ the pairs of neighbor profiles $(\x,\y)$ such that $\x \in A_1$ and $\y \in A_2$.
 We have:
 \begin{align*}
  \Dir_P(\varphi_L) & = \frac{\pi(L)^2}{2} \left(\sum_{(\x,\y) \in E(L, S \setminus L)} \pi(\x) P(\x,\y) + \sum_{(\x,\y) \in E(S \setminus L, L)} \pi(\x) P(\x,\y) \right)\\
  & = \pi(L)^2 \sum_{\x \in \partial L} \pi(\x) \sum_{\begin{subarray}{c}\y \in S\setminus L \colon\\H(\x,\y) = 1\end{subarray}} P(\x,\y) = \pi(L)^2 Q(L,S \setminus L),
 \end{align*}
where we used the reversibility of $P$ in the penultimate equality. Hence, we have $\frac{\Dir_P(\varphi_L)}{\Expec{\pi}{\varphi_L^2}} = B(L)$. The lemma follows since $1 - \lambda^{\overline{L}}_{\max} \leq \frac{\Dir_P(\varphi_L)}{\Expec{\pi}{\varphi_L^2}}$ (see Lemma~\ref{lem:lambda_max} in Appendix).
\end{proof}
The second characterization may be proved in exactly the same way as a similar well-known characterization for the spectral gap of stochastic matrices (see, for example, Section~13.3.3 in~\cite{lpwAMS08}).
\begin{lemma}
 \label{lemma:lambda_max_bottle_ub}
 For any $\emptyset \neq L \subseteq S$,
 $$
  1 - \lambda^{\overline{L}}_{\max} \geq \frac{(B^L_\star)^2}{2}.
 $$
\end{lemma}

Finally, let us recall a couple of lemmata relating $\tauset{S \setminus L}$ and $\lambda^{\overline{L}}_{\max}$
and already stated in e.g. \cite{msFOCS09}.
\begin{lemma}
\label{lem:hittin_lambda_max_lb}
 For a reversible Markov chain with state space $S$, any $L \subseteq S$ and any $t$ it holds that
 $$
  \max_{\x \in L} \Prob{\x}{\tauset{S \setminus L} > t} \geq \exp\left(t \log \lambda^{\overline{L}}_{\max}\right).
 $$
\end{lemma}

\begin{lemma}
\label{lem:hittin_lambda_max_ub}
 For a reversible Markov chain with state space $S$, any $L \subseteq S$ and any $t$ it holds that
 $$
  \Prob{\x}{\tauset{S \setminus L} > t} \leq \exp\left(t \log \lambda^{\overline{L}}_{\max} + \frac{1}{2} \log \frac{1}{\pi_L(\x)} \right),
 $$
 where $\pi_L(\x)$ has been defined in~\eqref{eq:piL}.
\end{lemma}
Since the statement of Lemma~\ref{lem:hittin_lambda_max_ub} is slightly different from the ones found in previous literature,
we provide a proof in Appendix~\ref{sec:hitting} for sake of completeness.

The above lemmata represent the main ingredients to prove the following relations between bottleneck ratio and hitting time.
\begin{lemma}
\label{lemma:hitting_bottle_ub}
 Let $\G$ be a potential game with profile space $S$ and let $P$ be the transition matrix of the logit dynamics for $\G$. Then for finite $\beta$ and $L \subset S$, $L \neq \emptyset$, we have
 $$
  \min_{\x \in L} \Prob{\x}{\tauset{S \setminus L} \leq t} \leq t \cdot \frac{B(L)}{1 - B(L)}.
 $$
\end{lemma}
\begin{proof}We observe:
{\allowdisplaybreaks
 \begin{align*}
 \min_{\x \in L} \Prob{\x}{\tauset{S \setminus L} \leq t} & = 1 - \max_{\x \in L} \Prob{\x}{\tauset{\overline{L}} > t}\\
 \text{(by Lemma~\ref{lem:hittin_lambda_max_lb})} \qquad & \leq 1- \exp\left(t \log \lambda^{\overline{L}}_{\max}\right) \\
 & = 1- \exp\left(t \log (1 - (1 - \lambda^{\overline{L}}_{\max}))\right) \\
 \text{(since $1 - a \geq e^{-\frac{a}{1-a}}$)} \qquad & \leq  1- \exp\left(-t\frac{1 - \lambda^{\overline{L}}_{\max}}{\lambda^{\overline{L}}_{\max}}\right)\\
 \text{(by Lemma~\ref{lemma:lambda_max_bottle})} \qquad & \leq 1 - \exp\left(-t \cdot \frac{B(L)}{1 - B(L)}\right)\\
 \text{(since $1 - e^{-a} \leq a$)} \qquad & \leq  t \cdot \frac{B(L)}{1 - B(L)}.  \tag*{\qed}
 \end{align*}}
 \let\qed\relax
\end{proof}
Moreover, we have the following lemma.
\begin{lemma}
\label{lemma:hitting_bottle_lb}
 Let $\G$ be a potential game with profile space $S$ and $P$ be the transition matrix of the logit dynamics for $\G$. For $\beta>0$, $\emptyset \neq L \subset S$, $\x \in L$ and $0 < \varepsilon < 1$, we have
 $$
  \Time_{S \setminus L}^\varepsilon(\x) \leq (B^L_\star)^{-2} \left(\frac{2(1-\varepsilon)}{\varepsilon} + \log \frac{1}{\pi_L(\x)}\right),
 $$
 where $\pi_L(\x) = \frac{\pi(\x)}{\pi(L)}$ and $B^L_\star = \min_{\begin{subarray}{c}A \subseteq L \colon\\\pi(A) \leq 1/2 \end{subarray}} B(A)$.
\end{lemma}
\begin{proof}
 From Lemma~\ref{lem:hittin_lambda_max_ub} we know that the hitting time of $S \setminus L$ can be expressed as a function of the eigenvalues of the matrix $P_{\overline{L}}$. In particular, we have
\begin{align*}
 \Prob{\x}{\tauset{S \setminus L} > t} & \leq \exp\left(t \log \lambda^{\overline{L}}_{\max} + \frac{1}{2} \log \frac{1}{\pi_L(\x)} \right)\\
 \text{(since $1 - a \leq e^{-a}$)} \qquad & \leq \exp\left(-t\left(1 - \lambda^{\overline{L}}_{\max}\right) + \frac{1}{2} \log \frac{1}{\pi_L(\x)}\right)\\
 \text{(by Lemma~\ref{lemma:lambda_max_bottle_ub})} \qquad & \leq \exp\left[-\frac{1}{2} \left(t (B^L_\star)^2 - \log \frac{1}{\pi_L(\x)}\right)\right]\\
 \text{(since $e^{-a} \leq (1 + a)^{-1}$)} \qquad & \leq \left(1 + \frac{1}{2} \left(t (B^L_\star)^2 - \log \frac{1}{\pi_L(\x)}\right)\right)^{-1}.
 \end{align*}
 Thus, by setting $t = (B^L_\star)^{-2} \left(\frac{2(1-\varepsilon)}{\varepsilon} + \log \frac{1}{\pi_L(\x)}\right)$, we have $\Prob{\x}{\tauset{S \setminus L} > t} \leq \varepsilon$ and then $\Time_{S \setminus L}^\varepsilon(\x)$ is upper bounded by this value of $t$.
\end{proof}

\subsubsection{Bounding the hitting time}
\label{subsubsec:pmt_N_small_beta}
The following lemma turns out to be useful for proving fast hitting time of profiles not in $N$.
\begin{lemma}
\label{lemma:hitting}
Let $\G$ be an asymptotically well-behaved $n$-player potential game and fix $\beta \geq \beta_0, \varepsilon > 0$.
Let $p, q$ the functions generated by $\G$.
Then, for each $n$ sufficiently large, at the end of algorithm $\A_{p,q}$ on input $\G, \beta, \varepsilon$ and $n$ it holds that for each subset $L \subseteq N$ such that $\pi(L) \leq 1/2$, $B(L) \geq 1/p(n)$.
\end{lemma}
\begin{proof}
It is sufficient to prove that for each $R_i$ chosen by $\A_{p,q}$, its core $T_i$ is non-empty. Indeed, in this case, the algorithm ends only if no subset $L \subseteq N$ such that $\pi(L) \leq 1/2$ has $B(L) > 1/q(n)$. Then, since $\G$ is asymptotically well-behaved, from Lemma~\ref{lem:awd_classified} it follows that the last condition is equivalent to $B(L) \geq 1/p(n)$.

As for the non-emptiness of the core, Lemma~\ref{lemma:hitting_bottle_ub} implies that there exists at least one $\x \in R_i$ such that
\[
\Prob{\x}{\tauset{S\setminus R_i} \leq \tm^{R_i}(\varepsilon)} \leq \frac{\tm^{R_i}(\varepsilon) \cdot B(R_i)}{1-B(R_i)} \leq \varepsilon,
\]
where the last step holds for $n$ sufficiently large since $\tm^{R_i}$ is at most polynomial by Proposition~\ref{prop:mixing_restricted_chain} and $B(R_i)$ is at most the inverse of a super-polynomial by hypothesis.
\end{proof}

We are now ready to prove Proposition~\ref{prop:hitting_outN}
\begin{proof}[Proof of Proposition~\ref{prop:hitting_outN}.]
Consider the set of profiles $A_\star \subseteq N$ that minimizes $B\left(A\right)$ among all $A \subseteq N$ such that $\pi\left(A\right) \leq 1/2$.
By Lemma~\ref{lemma:hitting}, $B\left(A_\star\right) \geq 1/p(n)$.
Moreover, for each $n$ and each $\x \in N$, observe that
$$
 \log \frac{1}{\pi_N\left(\x\right)} \leq \log \frac{|S| e^{-\beta \Phi_{\min}}}{e^{-\beta \Phi_{\max}}} \leq \log \frac{e^{n \log m(n)} e^{-\beta \Phi_{\min}}}{e^{-\beta \Phi_{\max}}} = n \log m(n) + \beta \left(\Phi_{\max} - \Phi_{\min}\right),
$$
where $\Phi_{\max}$ and $\Phi_{\min}$ denote the maximum and minimum of the potential $\Phi$ overall possible strategy profiles.
Since $\Phi_{\max} - \Phi_{\min} \leq n \cdot \Delta(n)$ and $\beta \leq \rho(n) / \Delta(n)$, then
$$
 \log \frac{1}{\pi_N\left(\x\right)} \leq n \cdot \left(\log m(n) + \rho(n)\right) = \rho'(n),
$$
where, by assumption on $m$ and $\rho$, $\rho'$ is a function at most polynomial in its input.
Then, for every $\x \in N$, from Lemma~\ref{lemma:hitting_bottle_lb} it follows
$$
 \Time_{S \setminus N}^\varepsilon\left(\x\right) \leq \left(\frac{1}{B\left(A_\star\right)}\right)^2 \cdot \left(\frac{2(1-\varepsilon)}{\varepsilon} + \log \frac{1}{\pi_{N}\left(\x\right)}\right) \leq p(n)^2 \cdot \left(\frac{2(1-\varepsilon)}{\varepsilon} + \rho'(n)\right) = \rho_\star(n),
$$
where $\rho_\star$ is a function at most polynomial in its input, since $p$ and $\rho'$ are.
\end{proof}

\section{Conclusions and open problems}
In this work we prove that for any asymptotically well-behaved potential game and any starting point of this game
there is a distribution that is metastable for super-polynomial time and it is quickly reached.
In the proof we also give a sufficient condition for a game (not necessarily potential)
to enjoy this metastable behavior.
It is a very interesting open problem to prove that this property is also \emph{necessary}.
The main obstacle for this direction consists of the fact that
we do not know any tool for proving or disproving metastability of 
distributions that are largely different from the ones considered in this work. Also, given that our arguments are game-independent, it would be interesting to see whether sufficient and necessary conditions can be refined for specific subclass of games.

Our convergence rate results hold if $\beta$ is small enough. As we mention above, an assumption on $\beta$ is in general necessary because when $\beta$ is high enough logit dynamics roughly behaves as best-response dynamics. Moreover, in this case, the only metastable distributions have to be concentrated around the set of Nash equilibria.
This is because for $\beta$ very high, it is extremely unlikely that a player leaves a Nash equilibrium.
Then, the hardness results about the convergence of best-response dynamics for potential games, cf. e.g. \cite{FPT04}, imply that the convergence to metastable distributions for high $\beta$ is similarly computationally hard. Interestingly, this difference in the behavior of the logit dynamics for different values of $\beta$ suggests that ``the more noisy the system is, the more (meta)stable it is.''

Our result is in a sense existential, since it is unpractical to explicitly describe the distributions via the execution of Algorithm~\ref{algo:meta_sets}. It is then an interesting open problem to characterize the sets $R_i$'s and $T_i$'s returned by this algorithm for some specific class of games in order to understand better the stability guarantee of the distributions.
A better understanding of spectra of the transition matrix along the lines of the results we prove in Section~\ref{sec:spectral} may help in answering some of the questions above.

Naturally, there are other questions of general interest about metastability that we do not consider. For example, akin to price of anarchy and price of stability, one may ask what is the performance of a system in a metastable distribution? One might also want to investigate metastable behavior of different dynamics in potential games, such as best-response dynamics. However, in the latter case, no matter what selection rule is used to choose which player has to move next, a profile is never visited twice in time since at each step the potential goes down. Therefore, the ``transient'' behavior of best-response dynamics would roughly correspond to a (possibly exponentially long) sequence of profiles visited. This, however, would not add much to our understanding of the transient phase of best-response dynamics.

\section*{Acknowledgments} We wish to thank Paul W. Goldberg for many invaluable discussions related to a number of results discussed in this paper, and an anonymous reviewer for the enlightening comments on an earlier version of this work.

\bibliographystyle{plain}
\bibliography{meta4LD}

\begin{thebibliography}{10}

\bibitem{Ferrer2010}
C.~Al\'os-Ferrer and N.~Netzer.
\newblock The logit-response dynamics.
\newblock {\em Games and Economic Behavior}, 68(2):413 -- 427, 2010.

\bibitem{asWINE09}
A.~Asadpour and A.~Saberi.
\newblock On the inefficiency ratio of stable equilibria in congestion games.
\newblock In {\em Proc. of the 5th International Workshop on Internet and
  Network Economics (WINE'09)}, volume 5929 of {\em Lecture Notes in Computer
  Science}, pages 545--552. Springer, 2009.

\bibitem{afpppSPAA11j}
V.~Auletta, D.~Ferraioli, F.~Pasquale, P.~Penna, and G.~Persiano.
\newblock Convergence to equilibrium of logit dynamics for strategic games.
\newblock {\em CoRR}, abs/1212.1884, 2012.
\newblock Preliminary version appeared in SPAA 2011.

\bibitem{afppSODA12j}
V.~Auletta, D.~Ferraioli, F.~Pasquale, and G.~Persiano.
\newblock Metastability of logit dynamics for coordination games.
\newblock {\em CoRR}, abs/1107.4537, 2011.
\newblock Preliminary version appeared in SODA 2012.

\bibitem{afppSAGT10}
V.~Auletta, D.~Ferraioli, F.~Pasquale, and G.~Persiano.
\newblock Mixing time and stationary expected social welfare of logit dynamics.
\newblock {\em Theory of Computing Systems}, pages 1--38, 2013.

\bibitem{bkmp2005}
N.~Berger, C.~Kenyon, E.~Mossel, and Y.~Peres.
\newblock Glauber dynamics on trees and hyperbolic graphs.
\newblock {\em Probability Theory and Related Fields}, 131:311--340, 2005.

\bibitem{blumeGEB93}
L.~E. Blume.
\newblock The statistical mechanics of strategic interaction.
\newblock {\em Games and Economic Behavior}, 5:387--424, 1993.

\bibitem{bovICM06}
A.~Bovier.
\newblock Metastability: a potential theoretic approach.
\newblock In {\em Proc. of the {I}nternational {C}ongress of {M}athematicians},
  volume III, pages 499--518. European Mathematical Society, 2006.

\bibitem{cdtFOCS06}
X.~Chen, X.~Deng, and S.~Teng.
\newblock Computing nash equilibria: Approximation and smoothed complexity.
\newblock In {\em FOCS}, pages 603--612, 2006.

\bibitem{csSODA07}
S.~Chien and A.~Sinclair.
\newblock Convergence to approximate nash equilibria in congestion games.
\newblock In {\em SODA}, pages 169--178, 2007.

\bibitem{dasSODA11}
C.~Daskalakis.
\newblock On the complexity of approximating a nash equilibrium.
\newblock In {\em SODA}, pages 1498--1517, 2011.

\bibitem{ding2009censored}
J.~Ding, E.~Lubetzky, and Y.~Peres.
\newblock Censored glauber dynamics for the mean field ising model.
\newblock {\em Journal of Statistical Physics}, 137(3):407--458, 2009.

\bibitem{ding2009mixing}
J.~Ding, E.~Lubetzky, and Y.~Peres.
\newblock The mixing time evolution of glauber dynamics for the mean-field
  ising model.
\newblock {\em Communications in Mathematical Physics}, 289(2):725--764, 2009.

\bibitem{ellisonECO93}
G.~Ellison.
\newblock Learning, local interaction, and coordination.
\newblock {\em Econometrica}, 61(5):1047--1071, 1993.

\bibitem{FPT04}
A.~Fabrikant, C.~H. Papadimitriou, and K.~Talwar.
\newblock The complexity of pure nash equilibria.
\newblock In {\em STOC}, pages 604--612, 2004.

\bibitem{holSPA04}
F.~Hollander.
\newblock Metastability under stochastic dynamics.
\newblock {\em Stochastic Processes and their Applications}, 114(1):1--26,
  2004.

\bibitem{holLNM09}
F.~Hollander.
\newblock Three lectures on metastability under stochastic dynamics.
\newblock In {\em Methods of Contemporary Mathematical Statistical Physics},
  volume 1970 of {\em Lecture Notes in Mathematics}, pages 1--24. Springer
  Berlin / Heidelberg, 2009.

\bibitem{hj90}
R.~A. Horn and C.~R. Johnson.
\newblock {\em Matrix Analysis}.
\newblock Cambridge University Press, 1990.

\bibitem{kmUAI02}
M.~J. Kearns and Y.~Mansour.
\newblock Efficient nash computation in large population games with bounded
  influence.
\newblock In {\em UAI}, pages 259--266, 2002.

\bibitem{llp2010}
D.~Levin, M.~Luczak, and Y.~Peres.
\newblock Glauber dynamics for the mean-field ising model: cut-off, critical
  power law, and metastability.
\newblock {\em Probability Theory and Related Fields}, 146:223--265, 2010.

\bibitem{lpwAMS08}
D.~Levin, Y.~Peres, and E.~L. Wilmer.
\newblock {\em Markov Chains and Mixing Times}.
\newblock American Mathematical Society, 2008.

\bibitem{lmmEC03}
R.~J. Lipton, E.~Markakis, and A.~Mehta.
\newblock Playing large games using simple strategies.
\newblock In {\em ACM Conference on Electronic Commerce}, pages 36--41, 2003.

\bibitem{MS96}
D.~Monderer and L.~S. Shapley.
\newblock Potential games.
\newblock {\em Games and Economic Behavior}, 14(1):124 -- 143, 1996.

\bibitem{msFOCS09}
A.~Montanari and A.~Saberi.
\newblock Convergence to equilibrium in local interaction games.
\newblock In {\em Proc. of the 50th Annual Symposium on Foundations of Computer
  Science (FOCS'09)}. IEEE, 2009.

\bibitem{ovCUP05}
E.~Olivieri and M.~E. Vares.
\newblock {\em Large deviation and metastability}.
\newblock Cambridge University Press, 2005.

\bibitem{Osn05}
S.~Osnaga.
\newblock On rank one matrices and invariant subspaces.
\newblock {\em Balkan Journal of Geometry and Its Applications}, 10(1):145,
  2005.

\bibitem{youngTR00}
H.~Peyton~Young.
\newblock {\em The diffusion of innovations in social networks}, chapter in
  ``The Economy as a Complex Evolving System'', vol. III, Lawrence E. Blume and
  Steven N. Durlauf, eds.
\newblock Oxford University Press, 2003.

\bibitem{R73}
R.~W. Rosenthal.
\newblock A class of games possessing pure-strategy nash equilibria.
\newblock {\em International Journal of Game Theory}, 2(1):65--67, 1973.

\bibitem{bookIrrationality}
R.~J. Schiller.
\newblock {\em Irrational Exuberance}.
\newblock Wiley, 2000.

\bibitem{svSTOC08}
A.~Skopalik and B.~V{\"o}cking.
\newblock Inapproximability of pure nash equilibria.
\newblock In {\em STOC}, pages 355--364, 2008.

\end{thebibliography}

\newpage

\appendix

\section{Hitting time tools}\label{sec:hitting}
Consider a reversible Markov chain with state space $S$ and transition matrix $P$. For $L \subseteq S$ let $P_{\overline{L}}$, $\lambda^{\overline{L}}_i$ and $\lambda^{\overline{L}}_{\max}$ as defined in Section~\ref{sec:meta_preliminaries}.
Here we give a well known (see, e.g., ~\cite{msFOCS09}) variational characterization of $\lambda^{\overline{L}}_{\max}$ as expressed by the following lemma.
\begin{lemma}
 \label{lem:lambda_max}
 Consider a reversible Markov chain with state space $S$, transition matrix $P$ and stationary distribution $\pi$. For any $L \subseteq S$ we have
 $$
  1 - \lambda^{\overline{L}}_{\max} = \inf_\varphi \frac{\Dir_P(\varphi)}{\Expec{\pi}{\varphi^2}},
 $$
 where
 $
  \Dir_P(\varphi) \text{ is defined as in~\eqref{eq:def:dir}} \text{, } \Expec{\pi}{\varphi^2} = \sum_\x \pi(\x) \varphi^2(\x)
 $
 and the inf is taken over functions $\varphi$ such that $\varphi(\x) = 0$ for $\x \in S \setminus L$ and $\Expec{\pi}{\varphi^2} \neq 0$.
\end{lemma}

Since the statement of Lemma~\ref{lem:hittin_lambda_max_ub} is slightly different from the ones found in previous literature,
we attach a proof for sake of completeness.
\begin{proof}[Proof of Lemma~\ref{lem:hittin_lambda_max_ub}]
 Let $\varphi_{L}$ be the characteristic function on $L$, that is $\varphi_{L}(\x) = 1$ if $\x \in L$ and $0$ otherwise. Then
 \begin{equation}
  \label{eq:hittin_ub_1}
  \Prob{\x}{\tauset{S \setminus L} > t} = \sum_{\y \in S} P_{\overline{L}}^t(\x,\y) = \sum_{\y \in S} P_{\overline{L}}^t(\x,\y) \varphi_{L}(\y) = (P_{\overline{L}}^t \varphi_{L})(\x).
 \end{equation}
 Since $P_{\overline{L}}$ is reversible with respect to $\pi_L$, we have that its eigenvectors, $\psi_1, \ldots, \psi_{|S|}$, form an orthonormal basis with respect to the inner product $\langle \cdot, \cdot \rangle_{\pi_L}$: in particular we can write $\varphi_{L} = \sum_{i} \alpha_i \psi_i$, where $\sum_i \alpha_i = 1$ and each $\alpha_i \geq 0$. Hence and from the linearity of the inner product we have
\begin{equation}
  \label{eq:hittin_ub_2}
 \begin{split}
  \langle P_{\overline{L}}^t \varphi_{L}, P_{\overline{L}}^t \varphi_{L} \rangle_{\pi_L} & = \sum_i \sum_j \langle \alpha_i \left(\lambda^{\overline{L}}_{i}\right)^t \psi_i, \alpha_j \left(\lambda^{\overline{L}}_{j}\right)^t \psi_j \rangle_{\pi_L}\\
  \mbox{(by orthogonality)} & = \sum_i \left(\lambda^{\overline{L}}_{i}\right)^{2t} \langle \alpha_i \psi_i, \alpha_i \psi_i \rangle_{\pi_L}\\
  & \leq \left(\lambda^{\overline{L}}_{\max} \right)^{2t} \langle \varphi_{L}, \varphi_{L} \rangle_{\pi_L} = \left(\lambda^{\overline{L}}_{\max} \right)^{2t},
 \end{split}
\end{equation}
 where the last equality follow from the definition of $\varphi_{L}$. Moreover,
 \begin{equation}
  \label{eq:hittin_ub_3}
  \pi_L(\x) [(P_{\overline{L}}^t \varphi_{L})(\x)]^2 \leq \sum_{\y \in S} \pi_L(\y) [(P_{\overline{L}}^t \varphi_{L})(\y)]^2 = \langle P_{\overline{L}}^t \varphi_{L}, P_{\overline{L}}^t \varphi_{L} \rangle_{\pi_L}.
 \end{equation}
 The theorem follows from~\eqref{eq:hittin_ub_1}, \eqref{eq:hittin_ub_2}, \eqref{eq:hittin_ub_3}.
\end{proof}

\section{Markov chain coupling} \label{sec:Markovcoupling}
A {\em coupling} of two probability distributions $\mu$ and $\nu$ on a state space $S$ is a pair of random variables $(X,Y)$ defined on $S \times S$ such that the marginal distribution of $X$ is $\mu$ and the marginal distribution of $Y$ is $\nu$. A {\em coupling of a Markov chain} $\mathcal{M}$ on $S$ with transition matrix $P$ is a process $(X_t,Y_t)_{t=0}^\infty$ with the property that $X_t$ and $Y_t$ are both Markov chains with transition matrix $P$. Similarly, a {\em coupling of Markov chains} $\mathcal{M}$, $\mathcal{\bar M}$ both defined on $S$ with transition matrices $P$ and $\bar P$, respectively, is a process $(X_t,Y_t)_{t=0}^\infty$ with the property that $X_t$ is a Markov chain with transition matrix $P$ and $Y_t$ is a Markov chain with transition matrix $\bar P$.

When the two coupled chains start at $(X_0,Y_0) = (\x,\y)$, we write $\Prob{\x,\y}{\cdot}$ for the probability of an event on the space $S \times S$.
The following theorem, which follows from Proposition~4.7 and Theorem~5.2 in \cite{lpwAMS08} establishes the importance of this tool.
\begin{theorem}[Coupling]
\label{thm:coupling}
Let $\mathcal{M}$, $\mathcal{\bar M}$ be two Markov chains with finite state space $S$ and
transition matrices $P$ and $\bar P$, respectively. For each pair of states $\x,\y \in S$ consider a coupling $(X_t,Y_t)$ of $\mathcal{M}$ and $\mathcal{\bar M}$ with starting states $X_0=\x$ and $Y_0=\y$.
Then
$$
\tv{P^t(\x,\cdot) - \bar{P}^t(\y,\cdot)} \leq \Prob{\x,\y}{X_t \neq Y_t}.
$$
\end{theorem}

\section{Asymptotic well-classified games}\label{sec:awc}
It is natural to ask if asymptotic metastability can hold even for potential games that are not asymptotically well-behaved.
Lemma~\ref{lem:potg_no_asymMeta} states that this is not the case in general,
but we wonder if there are conditions on potential games that are sufficient for asymptotic metastability.
In this section 
we show that the framework described above
can give a slightly more general condition than being asymptotically well-behaved.

To get an intuition of the condition that we are going to define, 
it is worth looking more closely at the game of Lemma~\ref{lem:potg_no_asymMeta}, known to be not partitioned by the logit dynamics (cf. Lemma~\ref{lem:potg_no_partitionable}). This game necessitates the update of the function $\Time$ infinitely often when adding a new player.
This update is done so that the new profile $(1, \ldots, 1)$ has a bottleneck ratio that
cannot be described by any of the functions considered at that point.
This process is never ending and gives no asymptotic meaning to $\Time$ in the limit.

The intuition is then that a game cannot be partitioned by the logit dynamics when for each choice of a polynomial $p$,
a super-polynomial $q$ and infinitely many values of $n$, there is a subset of states $L$ with $n$ players such that $\pi(L) < 1/2$ and
\[
p(n) < B^{-1}(L) < q(n).
\]
A sufficient property can then assume that we can asymptotically classify the bottleneck ratio of each subset of profiles as either polynomial or super-polynomial. More specifically, we can assume there are two functions $p$ at most polynomial in the input and $q$ at least super-polynomial in the input such that for each $\beta$ and each subset $A$ the bottleneck ratio $B(A)$ can be bounded by functions that depends on either $p$ or $q$. 
An equivalent viewpoint would be to see $n$-player potential games as a class to which a kind of oracle is attached that distinguishes between polynomial and super-polynomial bottleneck ratios for any fixed $\beta$. Formally, given a $n$-player potential game $\G$ and fixed $\beta \geq \beta_0$, this oracle can be described as follows: when it is queried about the bottleneck ratio of a subset $A$ with $n$ players its answer states that the bottleneck ratio is either i) at most polynomial if it is lower-bounded by $1/p(n)$; or ii) at least super-polynomial if it is upper-bounded by $1/q(n)$.

Lemma~\ref{lem:awd_classified} actually proves that this property 
holds for asymptotically well-behaved games. However, a more careful look to the game of Lemma~\ref{lem:potg_no_asymMeta} highlights that
to prove asymptotic metastability of a game with respect to a polynomial $p$ and a super-polynomial $q$
we do not need the behavior of each subset to be classified.
That is, we can allow some subsets of profiles to have bottleneck ratio in between the inverse of $q$ and the inverse of $p$.
In this case, we will say that the subset of profiles is \emph{unclassified}.

Our next condition then describes which class of subsets is sufficient to classify in order to have that the sets
returned by the Algorithm~\ref{algo:meta_sets} enjoy the properties required by the definition of partitioned games.
In particular, we define the class of \emph{asymptotically well-classified games} as follows.
\begin{definition}[Asymptotically well-classified games]
\label{def:awd}
 An $n$-player potential games $\G$ is asymptotically well-classified if for every $\beta \geq \beta_0, \varepsilon > 0$ there exist a pair of functions $p$ at most polynomial and $q$ at least super-polynomial, that for each $n$ sufficiently large, satisfy the following conditions:
 \begin{enumerate}
 \item $q(n) \leq \max_{L \colon \pi(L) \leq 1/2} B^{-1}(L)$;
   \item for each $R_i$ computed by $\A_{p,q}$ and for any $L \subset R_i$ such that $\pi_{R_i}(L) \leq 1/2$, if $B_{R_i}(L) < 1/p(n)$, then both $B(L)$ and $B(R_i \setminus L)$ are not unclassified;
  \item for each subset $L \subseteq N$, $N$ being as at the end of the algorithm $\A_{p,q}$, such that $\pi(L) \leq 1/2$, $B(L)$ is not unclassified.
 \end{enumerate}
\end{definition}
By careful looking at their proofs, one can check that Proposition~\ref{prop:mixing_restricted_chain} and Proposition~\ref{prop:hitting_outN} continue to hold even if we substitute asymptotically well-behaved potential games with asymptotically well-classified ones.

\section{Spectral properties of the logit dynamics}
\label{sec:spectral}
We next give other interesting spectral results about the transition matrix generated by the logit dynamics. In particular, by using a matrix decomposition similar to the one adopted in the proof of Lemma~\ref{lemma:diod_conj_restricted} we can prove the following propositions. (We remark that results in this section do not need to assume that the chain is reversible and indeed apply to any strategic game and not only to potential games.)
\begin{prop}
\label{prop:trace}
 Let $\mathcal{G}$ be a game with profile space $S$ and let $P$ be the transition matrix of the logit dynamics for $\mathcal{G}$. The trace of $P$ is independent of $\beta$.
\end{prop}
\begin{proof}
 For every $i$ and for every $\z_{-i}$ consider the transition matrices $P_{i,\z_{-i}}$ defined in~\eqref{eq:Pizi}, with $L = S$. Let $S_{i,\z_{-i}} = \{(\z_{-i}, s_i) \mid s_i \in S_i\}$. Observe that for every $\x \in S_{i,\z_{-i}}$ we have $P_{i,\z_{-i}}(\x,\x) = 1 - \sum_{\y \in S_{i,\z_{-i}}, \y \neq \x} P(\x,\y)$. Hence, the trace of $P_{i,\z_{-i}}$ is
 $$
  \sum_{\x \in S_{i,\z_{-i}}} P_{i,\z_{-i}}(\x,\x) = |S_i| - \sum_{\x \in S_{i,\z_{-i}}}\sum_{\y \in S_{i,\z_{-i}}, \y \neq \x} P(\x,\y).
 $$
 Since all non-zero elements in a column of $P_{i,\z_{-i}}$ are the same we also have
 $$
  P_{i,\z_{-i}}(\x,\x) = \frac{1}{|S_i|-1} \sum_{\y \in S_{i,\z_{-i}}, \y \neq \x} P(\y,\x).
 $$
 By setting $C = \sum_{\x \in S_{i,\z_{-i}}}\sum_{\y \in S_{i,\z_{-i}}, \y \neq \x} P(\x,\y) = \sum_{\x \in S_{i,\z_{-i}}}\sum_{\y \in S_{i,\z_{-i}}, \y \neq \x} P(\y,\x)$, we have
 $$
  |S_i| - C = \frac{C}{|S_i| - 1} \Longrightarrow C = |S_i| - 1,
 $$
 and thus, the trace of $P_{i,\z_{-i}}$ is always $1$, regardless of $\beta$. The proposition follows since the trace of $P$ is exactly the sum of the traces of all $P_{i,\z_{-i}}$'s.
\end{proof}
The proposition above says that if there exists an eigenvalue of $P$ that gets closer to 1 as $\beta$ increases, then there are other eigenvalues that get smaller: this is very promising in the tentative to characterize the entire spectrum of eigenvalues of $P$, necessary to use powerful tools such as the well-known \emph{random target lemma} \cite{lpwAMS08}.

In order to prove our last characterization of the transition matrix generated by the logit dynamics, we prove the following lemma which gives a lower bound on the probability that the strategy profile is not changed in one step of the logit dynamics for a generic game.
\begin{lemma}
\label{lemma:bound_loop}
 Let $\mathcal{G}$ be a game with profile space $S$ and let $P$ be the transition matrix of the logit dynamics for $\mathcal{G}$. Then for every $\x \in S$ we have that
 $$
  P(\x, \x) = \sum_i P\Big((\x_{-i},s_i^\star), \x\Big),
 $$
 where $s_i^\star \neq x_i$ is an arbitrary strategy of player $i$.
\end{lemma}
\begin{proof}
Observe that
 \begin{align*}
 P(\x,\x) & =  1 - \sum_{\y \in N(\x)} P(\x, \y) = \sum_i \left( \frac{1}{n} - \sum_{\y \in N_i(\x)} P(\x,\y)\right)\\
 & = \sum_i \frac{1}{n} \left( 1 - \sum_{\y \in N_i(\x)} \frac{e^{\beta u_i(\y)}}{e^{\beta u_i(\x)} + \sum_{\z \in N_i(\x)} e^{\beta u_i(\z)}}\right)
  =  \sum_i \frac{1}{n} \frac{e^{\beta u_i(\x)}}{e^{\beta u_i(\x)} + \sum_{\z \in N_i(\x)} e^{\beta u_i(\z)}}.
 \end{align*}
 The proof concludes by observing that for every $i$ and for every  $s_i^\star \in S_i$, we have
\[
  P\Big((\x_{-i}, s_i^\star), \x\Big) = \frac{1}{n} \frac{e^{\beta u_i(\x)}}{e^{\beta u_i(\x)} + \sum_{\z \in N_i(\x)} e^{\beta u_i(\z)}}.\tag*{\qed}
\]
\let\qed\relax
\end{proof}
Lemma~\ref{lemma:bound_loop} allows us to calculate the determinant of $P$.
\begin{prop}
 \label{prop:determinant}
Let $\mathcal{G}$ be a game with profile space $S$ and let $P$ be the transition matrix of the logit dynamics for $\mathcal{G}$. The  determinant of $P$ is $0$.
\end{prop}
\begin{proof}
It is well-known that a matrix in which one row can be expressed as a linear combination of other rows has determinant zero. In this proof, we fix a profile $\x$ and show that the row of $P$ corresponding to $\x$ can be obtained as a linear combination of other rows of the matrix. For each player $i$, fix a strategy $s_i^\star \in S_i$ such that $s_i^\star \neq x_i$. Let us denote with $S^j$, $j = 0, \ldots, n$, the set of profiles $\y \in S$ obtained from $\x$ by selecting $j$ players $i_1, \ldots, i_j$ and setting their strategies to $s_{i_1}^\star, \ldots, s_{i_j}^\star$, respectively. Notice that $\x$ belongs to $S^0$. By construction, for every profile $\z \in S^j$, $z_i \in \{x_i, s^{\star}_i\}$. Now, for $i=1, \ldots, n$, consider the profile obtained from $\z$ by changing $z_i=x_i$ into $s^{\star}_i$ or viceversa. Note that there are $n$ of such profiles which are neighbors of $\z$ and all contained in the sets $S^{j-1}$ and $S^{j+1}$. We claim that for every $\y\in S$
 \begin{equation}
 \label{eq:det_eq}
  P(\x,\y) = \sum_{j=1}^n (-1)^{j+1} \sum_{\z \in S^j} P(\z,\y).
 \end{equation}
 In order to prove the claim we distinguish three cases:
 \begin{enumerate}
  \item Let $H(\x,\y) > 1$ (and thus $P(\x,\y) = 0$): if there exists $j \in \{0, \ldots, n\}$ such that $\y \in S^j$, then the r.h.s. of~\eqref{eq:det_eq} becomes $\pm \Big( P(\y,\y) - \sum_i P\big((\y_{-i}, s_i^\star), \y\big)\Big) = 0$, from Lemma~\ref{lemma:bound_loop}; if $\y \notin \bigcup_{j = 0}^n S^j$, then consider a profile $\z \in S^j$, for some $j=1, \ldots, n$, such that $\z$ differs from $\y$ only in the strategy of player $k$: if no such profile exists, then the r.h.s. of~\eqref{eq:det_eq} is 0; otherwise, let us assume w.l.o.g. $z_k = x_k$ (the case $z_k = s_k^\star$ can be managed similarly), then the profile $\z' = (\z_{-k},s_k^\star)$ is a neighbor of $\y$, belongs to the set $S^{j+1}$ and $P(\z,\y) = P(\z',\y)$: hence, this two profiles delete each other in the r.h.s. of~\eqref{eq:det_eq}, giving the aimed result.

  \item Let $\x,\y$ differ in the strategy adopted by the player $k$: if there exists $j \in \{0, \ldots, n\}$ such that $\y \in S^j$, then the r.h.s. of~\eqref{eq:det_eq} becomes $P(\y,\y) - \sum_{i \neq k} P\big((\y_{-i}, s_i^\star), \y\big) = P(\x,\y)$, from Lemma~\ref{lemma:bound_loop}; if $\y \notin \bigcup_{j = 0}^n S^j$, then, as above, all profiles in $\bigcup_{j = 0}^n S^j$ that differ from $\y$ only in one player $i \neq k$ delete each other in the r.h.s. of~\eqref{eq:det_eq}: thus, the only element that survives in the r.h.s. of~\eqref{eq:det_eq} is $P\big((\x_{-k},x_k),\y\big) = P(\x,\y)$.

  \item If $\x = \y$, then the r.h.s. of~\eqref{eq:det_eq} becomes $\sum_{i \neq k} P\big((\y_{-i}, s_i^\star), \y\big) = P(\x,\x)$, from Lemma~\ref{lemma:bound_loop}.\qedhere
\end{enumerate}
\end{proof}
Since, as observed above, the logit dynamics for potential games defines a reversible Markov chain, Lemma~\ref{lemma:diod_conj_restricted} and Proposition~\ref{prop:determinant} imply that the last eigenvalue of the logit dynamics for these games is exactly 0. (Note that in \cite{afpppSPAA11j} is only stated the last eigenvalue is non-negative.)
Moreover, from the proof above, it turns out that an eigenvector of such zero eigenvalue is given by the function $f \colon S \rightarrow \mathbb{R}$ defined as
$$
 f(\w) = \begin{cases}
          -1, & \mbox{if $\w \in S^j$ and $j$ is even;}\\
          1, & \mbox{if $\w \in S^j$ and $j$ is odd;}\\
          0, & \mbox{otherwise;}\\
         \end{cases}
$$
where the sets $S^j$'s are defined as in the above proof from some fixed profile $\x$.

\end{document}